\documentclass[11pt]{article}

\usepackage{amsmath, amsthm, amssymb, amsfonts}
\usepackage{a4wide}
\usepackage[pdftex]{graphicx}
\usepackage[english]{babel}

\newtheorem{theorem}{Theorem}

\newtheorem{corollary}[theorem]{Corollary}
\newtheorem{lemma}[theorem]{Lemma}
\theoremstyle{remark}

\theoremstyle{definition}
\newtheorem{definition}[theorem]{Definition}

\title{Three new complexity results for resource allocation problems}

\author{Bart de Keijzer (B.deKeijzer@student.tudelft.nl)}

\begin{document}
\maketitle
\thispagestyle{plain}

\begin{abstract}
We prove the following results for task allocation of indivisible resources:
\begin{itemize}
\item The problem of finding a leximin-maximal resource allocation is in $\mathsf{P}$ if the agents have $\max$-utility functions and atomic demands.
\item Deciding whether a resource allocation is Pareto-optimal is $\mathsf{coNP}$-complete for agents with (1-)additive utility functions.
\item Deciding whether there exists a Pareto-optimal and envy-free resource allocation is $\Sigma_2^p$-complete for agents with (1-)additive utility functions.
\end{itemize}
\end{abstract}

\section{Introduction}\label{section:introduction}
In this text we prove complexity bounds for various problems in the field of resource allocation.
These results come forth from an attempt to prove two open problems that were stated in the work of Bouveret, Lang et al (\cite{aamasallocation05} and \cite{BouveretLangJAIR08}). The problems are about resource allocation. In a resource allocation problem we have a set of agents (or alternatively, players) and a set of resources (or equivalently, goods, tasks, items, etc.). The goal is to allocate the resources to the agents such that some requirements are satisfied. These requirements may vary. In our case we are interested in finding {\it fair} allocations. The concept of fairness is not clear, and there are different criteria for deciding whether or not an allocation is fair. Two of these are {\it envy-freeness} and {\it leximin-maximality}. We will define these criteria (formally) later on. In the problems we consider, the resources are indivisible and a resource can not be shared by two or more agents. 

The two open problems of the aforementioned papers that we consider are:

\begin{enumerate}
\item In \cite{aamasallocation05}: The problem of finding a leximin-maximal resource allocation for agents with $\max$-utility functions and atomic demands is in $\mathsf{NP}$. Could it be that it's in $\mathsf{NPC}$ (i.e. $\mathsf{NP}$-complete), or is it perhaps in $\mathsf{P}$?
\item In \cite{BouveretLangJAIR08}: What is the complexity of deciding whether there exists a Pareto-efficient and envy-free resource allocation, when the agents have additive utility functions?
\end{enumerate}

Some of the more technical notions we just mentioned will be defined and explained later in this text. We do, however, assume that the reader is acquainted with computational complexity theory (especially the classes $\mathsf{P}$, $\mathsf{NP}$,  $\mathsf{coNP}$, and the classes of the polynomial hierarchy), the matching problem for bipartite graphs, logic, and the satisfiability problem.

The first of these two problems is part of a quite an extensive series of problems and subproblems. The authors show for all of these problems that they are either in $\mathsf{P}$ or in $\mathsf{NPC}$. The only problem for which it remained an open question whether it is in $\mathsf{P}$ or in $\mathsf{NPC}$ (or possibly in between) is this one, where the agents have $\max$-utility and a leximin-optimal allocation must be found. In section \ref{section:leximin} we fill in the last open question of this series: we give a polynomial time algorithm for finding such an allocation, hence we prove that this problem is in $\mathsf{P}$\footnote{Of course we're talking about complexity classes for {\it decision} problems here. In \cite{aamasallocation05}, only the decision variant of this problem is considered. An algorithm from the decision variant of this problem is easily obtained if we have an algorithm for the optimization variant.}. 

The second problem is also part of a collection of problems that the authors prove complete for various complexity classes. This particular problem is again the last open problem in this series. We prove in section \ref{section:eef-existence} that this problem is $\Sigma_2^p$-complete (a class in the second level of the polynomial hierarchy) by a reduction from the complement of the language $\forall\exists\mathsf{3CNF}$ (that is a restriction of the more well-known problem known as $\mathsf{2QSAT}_{\forall}$ or $\mathsf{2TQBF}_{\forall}$): a complete problem for $\Pi_2^p$, which is naturally the complement of $\Sigma_2^p$.

In the process of trying to prove the $\Sigma_2^p$-completeness of the second problem, we stumbled on another interesting result, namely that the problem of deciding whether an allocation of resources to agents is Pareto-efficient (also called: Pareto-optimal, efficient) is $\mathsf{coNP}$-complete for agents with additive utility functions. We will give this proof in section \ref{section:paretooptimal}. $\mathsf{coNP}$-completeness of this problem has already been proved in the case of agents with $\geq 2$-additive utility functions (implied from \cite{k-additive}), but not yet in the case of (1-)additive utility functions.

\section{Leximin-maximal allocations with $\max$-utility and atomic demands}\label{section:leximin}

In this section, first, we make some definitions. After that we define the problem. Finally we give a polynomial time algorithm to solve the problem.

\subsection{Preliminaries}\label{preliminaries}

We first define formally the problem to solve. In a resource allocation problem, a set of resources must be divided among a set of agents. Such a division of resources to agents we call an {\it allocation}. 

The allocation must satisfy a certain set of {\it constraints}. Each agent has preferences on bundles of resources it may receive. The way these preferences are represented varies from setting to setting. In our case we use a cardinal preference structure: We represent the extent to which an agent values the bundle of resources he gets as real numbers. See for example \cite{mara-survey} for examples of preference structures.

Formally, we use the following definition for resource allocation settings:

\begin{definition}[(Indivisible) resource allocation setting]\label{def:allocationproblem}
An indivisible resource allocation problem instance is a 5-tuple $\langle A, \mathcal{O}, U, \mathcal{C}, u_c \rangle$, where $A = \{a_1, \ldots, a_n\}$ is a set agents, $\mathcal{O} = \{o_1, \ldots, o_m\}$ is a finite set of resources. $U = \{ u_1 , \ldots , u_n \}$ is a set of utility functions, $u_i$ is the utility function of agent $a_i$. For all $u \in U$, $u: 2^{\mathcal{O}} \rightarrow \mathbb{R}$. $\mathcal{C}$ is a finite set of constraints, and $u_c$ is a collective utility function to be defined later.
\end{definition}

\begin{definition}[Allocation of indivisible resources]
Given a resource allocation problem setting $\langle A, \mathcal{O}, U, \mathcal{C}, u_c \rangle$, an allocation is a mapping $a: A \rightarrow 2^{\mathcal{O}}$. 
\end{definition}

\begin{definition}[Admissability of an allocation]
Given a resource allocation setting $\langle A, \mathcal{O}, U, \mathcal{C}, u_c \rangle$, an allocation $a$ is admissable if it satisfies  all constraints in $\mathcal{C}$.
\end{definition}

For the specific case of the resource allocation problem that we are interested in, there is only one constraint in $\mathcal{C}$, namely the {\it preemption constraint}. Also, we restrict ourselves to a special case of $\max$-utility functions. The definitions of these concepts are as follows.

\begin{definition}[Preemption constraint]
Given a resource allocation setting $\langle A, \mathcal{O}, U, \mathcal{C}, u_c \rangle$ and an allocation $a$, then $a$ satisfies the preemption constraint $c_{\mathsf{preempt}}$ iff $\forall i \in A : \forall j \in A : (j \not= i) \rightarrow (a(i) \cap a(j) = \varnothing)$. We write $a \vDash c_{\mathsf{preempt}}$.
\end{definition}

In words, the preemption constraint requires that an item is allocated to no more than one agent.

\begin{definition}[$\max$-utility function]\label{def:maxutil}
In a resource allocation setting $\langle A, \mathcal{O}, U, \mathcal{C}, u_c \rangle$, a utility function $u \in U$ is a $\max$-utility function if $u(\mathcal{O'} \in 2^{\mathcal{O}}) = \max\{d_u(o) | o \subseteq \mathcal{O'} \}$, where $d_u: 2^{\mathcal{O}} \rightarrow \mathbb{R}$.
\end{definition}

In words, a $\max$-utility function has an associated demand function $d$. The $\max$-utility of a set of resources $\mathcal{O'}$ is the subset of $\mathcal{O'}$ for which the demand is the highest. We are interested in the following special case of $\max$-utility functions

\begin{definition}[$\max$-utility function with atomic demands]
$u$ is a $\max$-utility function with atomic demands if $u$ is a $\max$-utility function as defined in definition \ref{def:maxutil}, and $d_u$ has an associated atomic demand set $D_u = \{r_i, \ldots, r_m\} \subset \mathbb{R}$ such that

\begin{equation*}
d_u(\mathcal{O'} \in 2^{\mathcal{O}}) = 
\begin{cases} 
r_i & \text{ if } \mathcal{O'} = \{o_i\} \text{ for } 1 \leq i \leq m \\ 0 & \text{ otherwise } 
\end{cases}.
\end{equation*}
\end{definition}

This means: agents only express demands for single resources. Their utility for a set of resources is the highest demand they have for each of the individual resources of that set. Note that a $\max$-utility function is completely represented by its associated atomic demand set.

Now we are ready to discuss the {\it collective utility} function mentioned in definition \ref{def:allocationproblem}. The purpose of the collective utility function $u_c$ is to express the quality of an allocation. For this we need to be able to compare the answers that $u_c$ gives for any two different allocations. This implies:

\begin{itemize}
\item $u_c: (A \rightarrow 2^{\mathcal{O}}) \rightarrow X$,
\item we need to specify $X$,
\item we need to define a transitive comparison relation $\prec_X$ over $X$.
\end{itemize}

In a lot of cases we can say for example $X = \mathbb{R}$ or $X = \mathbb{N}$. The comparison relation is then simply $\leq$.
This is the case for classical utilitarian collective utility functions or egalitarian collective utility functions \cite{mara-survey}.
For us, the relation is a bit more complex. We are concerned with {\it leximin-egalitarian} collective utility functions.

\begin{definition}[Leximin-egalitarian collective utility]
Given a resource allocation setting $\langle A, \mathcal{O}, U, \mathcal{C}, u_c \rangle$. $u_c: (A \rightarrow 2^{\mathcal{O}}) \rightarrow X$ is a leximin-egalitarian collective utility function iff $X = \mathbb{R}^n$ and for all allocations $a$: $u_c(al) = \vec{x}$, where 
\begin{equation*}
\vec{x} = \begin{bmatrix}u_1(al(1)) \\ \vdots \\ u_n(al(n)) \end{bmatrix}.
\end{equation*}
\end{definition}

\begin{definition}[Leximin-egalitarian comparison relation]
The leximin-egalitarian comparison relation $\prec_{\mathsf{leximin}}$ is defined as follows: Let $\vec{u} \in \mathbb{R}^n$ and $\vec{v} \in \mathbb{R}^n$ and let $\vec{u^{\uparrow}}$ and $\vec{v^{\uparrow}}$ be the sorted versions of $\vec{u}$ and $\vec{v}$ respectively. Now, it holds that
\begin{equation*}
\vec{v} \prec_{\mathsf{leximin}} \vec{u} \Leftrightarrow \exists i : \forall j < i: \vec{v_j^{\uparrow}} = \vec{u_j^{\uparrow}} \wedge \vec{v_i^{\uparrow}} < \vec{u_i^{\uparrow}}.
\end{equation*}
\end{definition}

\begin{definition}[Leximin-maximality]
Given a resource allocation setting $\langle A, \mathcal{O}, U, \mathcal{C}, u_c \rangle$, with $u_c$ being a leximin-egalitarian collective utility function. An admissable allocation $a$ is leximin-maximal if there exists no admissable allocation $a'$ such that $u_c(a) \prec_{\mathsf{leximin}} u_c(a')$.
\end{definition}

A leximin-maximal allocation has a desirable `fairness'-property to it: The most important priority in a leximin-maximal allocation, is that the lowest utility among all the agents is as high as possible. As a second most important priority, the second-lowest utility among all the agents is made as high as possible, etcetera.

Finally we are ready to state the problem that we will prove to be in $\mathsf{P}$.

\begin{definition}[LMMUAB-ALLOCATION (i.e. Leximin-maximal $\max$-utility atomic bids resource allocation)]
A problem instance of LMMUAB-ALLOCATION is a resource allocation problem setting $\langle A, \mathcal{O}, U, \mathcal{C}, u_c \rangle$ and a vector $K$, where
\begin{itemize}
\item $u_c$ is a leximin-egalitarian collective utility function, 
\item $\mathcal{C} = \{c_{\mathsf{preempt}}\}$,
\item $\forall u \in U : u$ is a $\max$-utility function with atomic demands.
\item $K \in \mathbb{R}^n$
\end{itemize}

It is sufficient to represent a LMMUAB-ALLOCATION-instance as the triple $\langle A, \mathcal{O}, D \rangle$, where $D = \{D_1, \ldots, D_n\}$ is a set of atomic demand sets, and for $1 \leq i \leq n$, $D_i$ is the atomic demand set associated with $a_i$ and $u_i$.

The task is to determine if there exists an admissable allocation $a$ such that 
\begin{equation*}
K \prec_{\mathsf{leximin}} u_c(a).
\end{equation*}

\end{definition}

We prove LMMUAB-ALLOCATION in $\mathsf{P}$ by giving a polynomial time algorithm for its optimization variant.

\begin{definition}[LMMUAB-ALLOCATION-OPT (i.e. Leximin-maximal $\max$-utility atomic bids resource allocation, optimization variant)]
A problem instance of LMMUAB-ALLOCATION-OPT is the same as a problem instance of LMMUAB-ALLOCATION, but without the vector $K$. The task is to find a leximin-maximal, admissable allocation.
\end{definition}

\subsection{A polynomial time algorithm for LMMUAB-ALLOCATION-OPT}

Consider the following algorithm for LMMUAB-ALLOCATION-OPT:

\begin{center}
\begin{tabular}{|rl|} \hline
 & \bf{Algorithm A:} \\
{\bf Input: }& $I$, an instance of LMMUAB-ALLOCATION-OPT. \\
 & That is, $I = \langle A = \{a_1, \ldots, a_n\}, \mathcal{O} = \{o_1, \ldots, o_m\}, D = \{D_1, \ldots, D_n\} \rangle $, \\
 & and for $1 \leq i \leq n, D_i = \{r_{i,1}, \ldots, r_{i,m}\}$.\\
{\bf Output: }& $a$, a leximin-maximal allocation for $I$.\\
{\bf Begin} & \\
{\bf 1.} & Create a complete weighted bipartite graph $G = (V = (L \cup R), E)$, \\
 & where $L$ and $R$ are the left and right parts of the graph respectively. \\
 & We set $L := \mathcal{O}, R := A$.\\
{\bf 2.} & Generate weights $\ell_{i,j}$ for all $\{a_i, o_j\} \in E$ such that \\
         & $\ell_{i,j} \geq \sum_{ \{ (i',j') | r_{i',j'} > r_{i,j} \} } \ell_{i',j'}$. \\
{\bf 3.} & Find with the Hungarian algorithm \cite{citeulike:1274573} a minimum weighted bipartite \\
 & matching $M$ on $G$, using the weights computed in step 2.\\
{\bf 4.} & For all ${i,j} \in M$, set $a(a_i) := \{o_j\}$.\\
{\bf End} & \\
\hline
\end{tabular}
\end{center}
First please note: a minimum weighted bipartite matching is a maximum matching in a weighted bipartite graph such that the {\it cumulative weight} of the matching (i.e. the sum of the weights of the edges in the matching) is minimal. See for example \cite{citeulike:472316}.

We will now prove that this algorithm is correct and runs in polynomial time. From these two facts it follows that the decision variant of this problem also runs in polynomial time and hence is in $\mathsf{P}$

\begin{theorem}\label{optimality}
Algorithm A is a correct algorithm for LMMUAB-ALLOCATION-OPT, i.e. the allocation that algorithm A outputs on an LMMUAB-ALLOCATION-OPT-instance as input, is leximin-maximal.
\end{theorem}

\begin{proof}
First note that there exists a leximin-maximal allocation in which every agent gets at most one resource. This is due to the combination of $\max$-utility functions with atomic demands: of a bundle allocated to an agent, only a single resource in that bundle decides the agent's utility of that bundle, so we could just as well remove all the other items from the bundle. 

Step 4 allocates an item to an agent if the corresponding edge is in $M$. Because $M$ is a minimum weighted matching, an agent is allocated at most 1 item. What remains is proving that if our algorithm has found a minimum weighted matching $M$, then the algorithm constructs a leximin-maximal $a$. Suppose that is not the case: call the leximin-maximal allocation $a_{\mathsf{OPT}}$, and assume our algorithm returns an $a$ such that $u_c(a) \prec_{\mathsf{leximin}} u_c(a_{\mathsf{OPT}})$. By the definition of the leximin order $\prec_{\mathsf{leximin}}$ this means that

\begin{equation*}
\exists i : \forall j < i: u_c(a)_j^{\uparrow} = u_c(a_{\mathsf{OPT}})_j^{\uparrow} \wedge u_c(a)_i^{\uparrow} < u_c(a_{\mathsf{OPT}})_i^{\uparrow}.
\end{equation*}

We will now prove that there exists not such an $i$, resulting in a contradiction. We prove by induction that for all $1 \leq i \leq n: u_c(a)_i^{\uparrow} = u_c(a_{\mathsf{OPT}})_i^{\uparrow}$. For the remainder of the proof, let $M_{\mathsf{OPT}}$ be the matching that corresponds to $a_{\mathsf{OPT}}$, in the same way as $M$ corresponds to $a$.

\paragraph{Base case} $u_c(a)_1^{\uparrow} = u_c(a_{\mathsf{OPT}})_1^{\uparrow}$.
First of all, by construction of the weights in step 3, for all $1 \leq i \leq n, 1 \leq i' \leq n, 1 \leq j \leq m, 1 \leq j' \leq m : r_{i, j} < r_{i',j'} \Leftrightarrow \ell_{i,j} > \ell_{i',j'}$. So the edge with highest weight in $M$ corresponds to the agent with the lowest utility of the allocation, hence this utility corresponds to $u_c(a)_1^{\uparrow}$.
Secondly, let $e$ and $e_{\mathsf{OPT}}$ be the edges with the highest weight that are in $M$ and $M_{\mathsf{OPT}}$ respectively. Now, consider the set of edges $E_{>}$ with weights that are strictly greater than the weight of $e_{\mathsf{OPT}}$. By construction of the weights, it follows that any matching in which an $e' \in E_{>}$ is included, always has a greater cumulative weight than a matching in which $e_{\mathsf{OPT}}$ is included as the edge with the highest weight. Step 4 of the algorithm returns the matching with minimum cumulative weight, so the weight of $e$ must be the weight of $e_{\mathsf{OPT}}$.

\paragraph{Induction hypothesis} $\forall j < i: u_c(a)_j^{\uparrow} = u_c(a_{\mathsf{OPT}})_j^{\uparrow}$.

\paragraph{Induction step} $u_c(a)_i^{\uparrow} = u_c(a_{\mathsf{OPT}})_i^{\uparrow}$.
This follows more or less trivially from the same arguments as given for the base case: let $e^i$ and $e_{\mathsf{OPT}}^i$ be the edges with the $i$'th highest weight that are in $M$ and $M_{\mathsf{OPT}}$ respectively. Now, consider the set of $i$'th highest edges $E_{>}^i$ with weights that are strictly greater than the weight of $e_{\mathsf{OPT}}^i$ and strictly less than the weight of edge $e_{\mathsf{OPT}}^{i-j}, 1 \leq j \leq n-1$. By construction of the weights, it follows that any matching in which an $e' \in E_{>}^i$ is included as an $i$'th highest edge, always has a greater cumulative weight than a matching in which $e_{\mathsf{OPT}}^i$ is included as an $i$'th highest edge. Step 3 of the algorithm returns the matching with minimum cumulative weight, so the weight of $e^i$ must be the weight of $e_{\mathsf{OPT}}^i$.

\end{proof}

\begin{theorem}\label{polynomiality}
Algorithm A runs in polynomial time.
\end{theorem}

\begin{proof} The complexities of the individual steps of the algorithm are\footnote{We assume a RAM-model where the elementary arithmetic operations take unit time.}:

\begin{itemize} 
\item In step 1, $m+n$ nodes and $mn$ edges are constructed. This takes $O(mn)$ time.
\item In step 2 $mn$ weights are computed. This step is not described in a very constructive way, but it can be easily seen that it can be done by first sorting the union of all the demand vectors, and subsequently constructing the weights from the highest to the lowest element in the sorted array. In this step, the sorting is the most intensive part and takes $O(mn \log mn)$ time.
\item In step 3 the Hungarian algorithm for minimum weighted bipartite matchings is ran. This algorithm needs a helper shortest-path algorithm. If we use Dijkstra's algorithm as a helper algorithm for the Hungarian algorithm, then this step can be done in $O((m+n) \log (m+n) + (m+n)(m^2n^2))$ time  \cite{citeulike:472316}.
\item Step 4 is clearly done in $O(m+n)$ time.
\end{itemize}

Adding up the complexities of these steps, we conclude that the algorithm can run in  $O((m+n) \log (m+n) + (m+n)(m^2n^2))$ time.
\end{proof}
\begin{corollary}[from theorems \ref{optimality} and \ref{polynomiality}]
LMMUAB-ALLOCATION is in $\mathsf{P}$.
\end{corollary}

\section{Complexity of deciding whether an allocation is pareto optimal for agents with additive utility}
\label{section:paretooptimal}

In this section we prove that deciding whether an allocation of resources among a set of agents is $\mathsf{coNP}$-complete if the agents have additive utility functions. We will make use of the definitions given in section \ref{preliminaries}.
As said in the introduction of this paper, $\mathsf{coNP}$-completeness has already been proved for the case where agents have $k$-additive utility functions and $k \geq 2$. 

\begin{definition}[$k$-additive utility]
In a resource allocation setting $\langle A, \mathcal{O}, U, \mathcal{C}, u_c \rangle$, a utility function $u_i$ of an agent $a_i$ is $k$-additive if for each set $T \subseteq \mathcal{O}$ with $|T| = k$ there exists a coefficient $\alpha_T$ and for all $R \subseteq \mathcal{O}$ it holds that 
\begin{equation*}
u_i(R) = \sum_{T \subseteq \mathcal{R}} \alpha_T.
\end{equation*}
\end{definition}

$k$-additive utility functions are a generalisation of {\it additive} utility functions.

\begin{definition}[additive utility]\label{definition:additiveutility}
An additive utility function is a $k$-additive utility function with $k = 1$, i.e. a 1-additive utility function. An additive utility function can be represented as a set of coefficients: one coefficient for each item in $\mathcal{O}$.
\end{definition}

Next, we define the notion of {\it Pareto-efficiency}.

\begin{definition}[Pareto-efficiency]
In a resource allocation setting $\langle A, \mathcal{O}, U, \mathcal{C}, u_c \rangle$, an admissable allocation $a$ is Pareto-efficient (also called: Pareto-optimal, or simply efficient) if there exists not a different admissable allocation $a'$ where the utility of at least one agent is higher than in allocation $a$, and the utilities of all other agents are not lower than in allocation $a$. More formal: allocation $a$ is Pareto-optimal if there exists no allocation $a'$ such that 
\begin{equation*}
\exists a_i \in A : u_i(a'(a_i)) > u_i(a(a_i)) \wedge (\forall a_j \in A : u_j(a'(a_j)) \geq u_j(a(a_j))).
\end{equation*}
If such an allocation $a'$ does exist, then $a$ is not Pareto-optimal and we say that $a'$ {\it Pareto-dominates} $a$.
Also we say that $a$ can be {\it Pareto-improved} to $a'$ if $a'$ is an allocation that Pareto-dominates $a$. The process of reallocating items to get from $a$ to $a'$ is called a {\it Pareto-improvement}. If for $a$ there is no Pareto-improvent possible, then clearly $a$ is Pareto-optimal.
\end{definition}

Now we state the problem and prove it $\mathsf{coNP}$-complete.

\begin{definition}[PO-ALLOCATION-ADDITIVE (i.e. Pareto-Optimal Allocation with Additive utility functions)]
A problem instance of PO-ALLOCATION-ADDITIVE is a resource allocation problem setting $\langle A, \mathcal{O}, U, \mathcal{C}, u_c \rangle$ and an associated admissable allocation $a : A \rightarrow 2^{\mathcal{O}}$, where
\begin{itemize}
\item $\mathcal{C} = \{c_{\mathsf{preempt}}\}$,
\item $\forall u \in U : u$ is an additive utility function.
\end{itemize}

The problem is to decide whether $a$ is Pareto-optimal.
The collective utility function $u_c$ can be disregarded here, so the problem is representable as the 4-tuple $\langle A, \mathcal{O}, V, a\rangle$. In this 4-tuple, $V = \{v_1, \ldots, v_n\}$ represents the utility functions of $U$. For all $1 \leq i \leq n$, $v_i$ is the representation of $u_i$ as described in definition \ref{definition:additiveutility}.
\end{definition}

\begin{theorem}\label{additive}
PO-ALLOCATION-ADDITIVE is $\mathsf{coNP}$-complete.
\end{theorem}
\begin{proof}
Showing membership of $\mathsf{coNP}$ is easy: If the allocation $a$ of a PO-ALLOCATION-ADDITIVE-instance is not Pareto-optimal, then a certificate would be an allocation that Pareto-dominates $a$.

Proving $\mathsf{coNP}$-hardness for this problem is very difficult. We do it by a Karp reduction from 3-UNSAT. 3-UNSAT is the problem of deciding whether a propositional formula in 3CNF is unsatisfiable. Because satisfiable instances of such a formula are easy to verify, the complement of 3-UNSAT is in $\mathsf{NP}$. Hence 3-UNSAT is in $\mathsf{coNP}$.

The reduction is as follows. We are given an instance of 3-UNSAT $I$ with variables $\{x_1, \ldots, x_w\}$ and clauses $\{c_1, \ldots, c_{w'}\}$. A clause is given as a set of at most 3 literals. We transform this instance to a PO-ALLOCATION-ADDITIVE instance $I'$ in the following way. As in the definition, $I'$ is represented as the 4-tuple $\langle A, \mathcal{O}, V, a \rangle$.

\begin{itemize}
\item In $I'$, $|A| = 2w + w' + 2$: For each variable $x_i$ in $I$, two agents are introduced: $a_{\mathsf{set}(x_i)}$ and $a_{\mathsf{set}(\neg x_i)}$. $a_{\mathsf{set}(x_i)}$ represents the set of clauses in which the literal $x_i$ occurs. $a_{\mathsf{set}(\neg x_i)}$ represents the set of clauses in which the literal $\neg x_i$ occurs. For each clause $c_i$ in $I$, one agent $a_{c_i}$ is introduced in $I'$. Lastly, 2 additional agents are introduced: $a_{\mathsf{unassigned}}$ and $a_{\mathsf{satisfied}}$.
\item In $I'$, $|\mathcal{O}| = w + w' + L + 1$, where $L$ is the total number of literals in the formula. For each clause $c_i$ we introduce for each literal $l$ in that clause the resource $o_{c_i, l}$. For each variable $x_i$ we introduce the resource $o_{x_i}$. For each clause $c_i$ we introduce the resource $o_{c_i}$. Lastly, the resource $o_{\mathsf{satisfied}}$ is added.
\item The additive utility functions $V$ of the agents are specified as follows. Remember that we use the following names:
\begin{eqnarray*}
V & =    & \{v_{\mathsf{set}(x_1)}, \ldots, v_{\mathsf{set}(x_w)}\} \\
  & \cup & \{v_{\mathsf{set}(\neg x_1)}, \ldots, v_{\mathsf{set}(\neg x_w)}\} \\
  & \cup & \{v_{c_1}, \ldots, v_{c_w'}\} \\
  & \cup & \{v_{\mathsf{unassigned}}, v_{\mathsf{satisfied}}\}.
\end{eqnarray*}

All $v \in V$ are vectors of coefficients. We name these coefficients as follows. Let $a_i \in A$, and let $o_j \in \mathcal{O}$. Thus, $i$ and $j$ stand not for numbers in this case, but for subscripts. Then the coefficient for resource $j$ in the additive utility function of agent $i$ goes by the name of $\alpha_{i,j}$ (and hence $\alpha_{i,j} \in v_i$).

The coefficients for all resources for all agents are set to zero, with the following exceptions:
\begin{itemize}
\item All coefficients in $\{\alpha_{\mathsf{unassigned}, x_1}, \ldots, \alpha_{\mathsf{unassigned}, x_w}\}$ are set to 1.
\item All coefficients in $\{\alpha_{\mathsf{satisfied}, c_1}, \ldots, \alpha_{\mathsf{satisfied}, c_{w'}}\}$ are set to 1.
\item All coefficients in $\{\alpha_{c_1, c_1}, \alpha_{c_2, c_2}, \ldots, \alpha_{c_{w'}, c_{w'}}\}$ are set to 1.
\item For all coefficients $\alpha_{\mathsf{set}(l), x_i}$ in 
\begin{eqnarray*} 
& & \{\alpha_{\mathsf{set}(x_1), x_1}, \alpha_{\mathsf{set}(x_2), x_2}, \ldots, \alpha_{\mathsf{set}(x_w), x_w}\} \\
& \cup & \{\alpha_{\mathsf{set}(\neg x_1), x_1}, \alpha_{\mathsf{set}(\neg x_2), x_2}, \ldots, \alpha_{\mathsf{set}(\neg x_w), x_w}\},
\end{eqnarray*}
$\alpha_{\mathsf{set}(l), x_i}$ is set to the number of times that $l$ occurs in the formula of $I$.
\item All coefficients in
\begin{equation*} 
\{\alpha_{\mathsf{set}(l), (c_i, l)} | 1 \leq i \leq w' \wedge l \in c_i \}
\end{equation*}
are set to 1.
\item
All coefficients in
\begin{equation*} 
\{\alpha_{c_i, (c_i, l)} | 1 \leq i \leq w' \wedge l \in c_i \}
\end{equation*}
are set to 1.
\item $\alpha_{\mathsf{satisfied}, \mathsf{satisfied}}$ is set to $w'$ and $\alpha_{\mathsf{satisfied}, \mathsf{unassigned}}$ is set to $w + 1$.
\end{itemize}
\item Lastly, we must specify the allocation $a$. 
\begin{itemize}
\item All resources $\{o_{x_1}, \ldots, o_{x_w}\}$ are allocated to $a_{\mathsf{unassigned}}$.
\item For all resources $o_{c_i}, 1 \leq i \leq w'$ we allocate $o_{c_i}$ to $a_{c_i}$.
\item All resources $o_{c_i, l}, 1 \leq i \leq w', l \in c_i$, are allocated to $a_{\mathsf{set}(l)}$.
\item The resource $o_{\mathsf{satisfied}}$ is allocated to agent $a_{\mathsf{satisfied}}$.
\end{itemize}
\end{itemize}

That completes the reduction. It can clearly be done in polynomial time. Before continuing with the correctness proof of this reduction, an example would be appropriate, due to the complexity of the reduction.

Consider the 3-UNSAT instance given by the formula 
\begin{equation*}
% x1x2!x3 	!x1!x2!x3
(x_1 \vee x_2 \vee \neg x_3) \wedge (\neg x_1 \vee \neg x_2 \vee \neg x_3).
\end{equation*}
We represent this instance as the tuple
\begin{eqnarray*}
\langle & \{x_1,x_2,x_3\}, & \\
        & \{c_1 = \{x_1,x_2,\neg x_3\}, c_2 = \{\neg x_1, \neg x_2, \neg x_3\} \}& \rangle
\end{eqnarray*}

Now if we run the reduction process on this instance, we get a PO-ALLOCATION-ADDITIVE instance that is displayed in the table below. The columns of the table represent the agents and the rows of the table represent the items. The entries in the table are the coefficients. An entry is displayed in italic if the item of the corresponding row is allocated to the agent of the corresponding column. Empty cells in the table should be regarded as zero entries.

\small{
\begin{center}
\begin{tabular}{r|r|r|r|r|r|r|r|r|r|r|}

           &    $a_{c_1}$ &    $a_{c_2}$ & $a_{\mathsf{set}(x_1)}$ & $a_{\mathsf{set}(\neg x_1)}$ & $a_{\mathsf{set}(x_2)}$ & $a_{\mathsf{set}(\neg x_2)}$ & $a_{\mathsf{set}(x_3)}$ & $a_{\mathsf{set}(\neg x_3)}$ & $a_{\mathsf{unassigned}}$ & $a_{\mathsf{satisfied}}$ \\
\hline
   $o_{x_1}$ &            &            &          1 &          1 &            &            &            &            &    {\it 1} &            \\
\hline
   $o_{x_2}$ &            &            &            &            &          1 &          1 &            &            &    {\it 1} &            \\
\hline
   $o_{x_3}$ &            &            &            &            &            &            &            &          2 &    {\it 1} &            \\
\hline
   $o_{c_1}$ &    {\it 1} &            &            &            &            &            &            &            &            &          1 \\
\hline
   $o_{c_2}$ &            &    {\it 1} &            &            &            &            &            &            &            &          1 \\
\hline
$o_{c_1,x_1}$ &          1 &            &    {\it 1} &            &            &            &            &            &            &            \\
\hline
$o_{c_1,x_2}$ &          1 &            &            &            &    {\it 1} &            &            &            &            &            \\
\hline
$o_{c_1,\neg x_3}$ &          1 &            &            &            &            &            &            &    {\it 1} &            &            \\
\hline
$o_{c_2,\neg x_1}$ &            &          1 &            &    {\it 1} &            &            &            &            &            &            \\
\hline
$o_{c_2,\neg x_2}$ &            &          1 &            &            &            &    {\it 1} &            &            &            &            \\
\hline
$o_{c_2,\neg x_3}$ &            &          1 &            &            &            &            &            &    {\it 1} &          &            \\
\hline
$o_{\mathsf{satisfied}}$ &            &            &            &            &            &            &            &            &          4 &    {\it 2} \\
\hline
\end{tabular}  
\end{center} 
}
\normalsize
Now we will continue with the correctness proof. We must show that there only exists a Pareto-dominating allocation if the formula of the 3-UNSAT instance is satisfiable. This follows from the following two lemmas and concludes the proof.

\begin{lemma}
If the 3-UNSAT instance $I$ is a NO-instance, i.e. the formula is satisfiable, then the allocation $a$ in $I'$ is not Pareto-optimal.
\end{lemma}
\begin{proof}
First have to explain the function of all agents and resources with respect to the 3-UNSAT instance $I$.
The allocations of resources $\{o_{x_1}, \ldots, o_{x_{w}}\}$ represent to which truth-value the variables are set. If $o_{x_i}$ is allocated to $a_{\mathsf{unassigned}}$, this means that $x_i$ is set to no truth-value. If $o_{x_i}$ is allocated to $a_{\mathsf{set}(x_i)}$, this means that $x_i$ is set to true, and the clauses in which the literal $x_i$ occurs are made true. If $o_{x_i}$ is allocated to $a_{\mathsf{set}(\neg x_i)}$, this means that $x_i$ is set to false, and clauses in which the literal $\neg x_i$ occurs are made true. The agents $\{a_{c_1}, \ldots, a_{c_{w'}}\}$ represent the clauses of the formula. If resource $o_{c_i} \in \{o_{c_1}, \ldots, o_{c_{w'}}\}$ is allocated to $a_{c_i}$, it means that clause $c_i$ is not satisfied. If resource $o_{c_i} \in \{o_{c_1}, \ldots, o_{c_{w'}}\}$ is allocated to $a_{\mathsf{satisfied}}$, it means that clause $c_i$ is satisfied. In allocation $a$, all clauses are unsatisfied and all variables are not assigned a truth-value. If in allocation $a$, we reallocate some $o_{x_i} \in \{x_1, \ldots, x_{w} \}$ to one of the agents $a_{\mathsf{set}(l_{x_i})} \in \{ a_{\mathsf{set}(x_i)} , a_{\mathsf{set}(\neg x_i)} \}$, then by construction we can move all of the resources $o_{c_j, l_{x_i}}, l_{x_i} \in c_j$ to $a_{c_j}$ without lowering the utility of $a_{\mathsf{set}(l_{x_i})}$. Now, because $c_j$ gets 1 extra utility, we are able to reallocate $o_{c_j}$ to $a_{\mathsf{satisfied}}$. 

The key thing to see here is that the procedure we just described is, from the viewpoint of $I$, equivalent to assigning $x_i$ some truth value, and making all clauses true in which the literal occurs that corresponds to that truth-value. In $I'$ this is the same as reallocating some specific resources to some specific agents, and this reallocation can be done without lowering anyone's utility except for the utility of $a_{\mathsf{unassigned}}$. The utility of $a_{\mathsf{unassigned}}$ can only be compensated if $a_{\mathsf{unassigned}}$ gets allocated the resource $o_{\mathsf{satisfied}}$. If that happens, then by construction the utility of $a_{\mathsf{unassigned}}$ gets suddenly strictly higher than in allocation $a$. But we can only reallocate $o_{\mathsf{satisfied}}$ to $a_{\mathsf{unassigned}}$ if all resources $\{o_{c_1}, \ldots, o_{c_{w'}} \}$ are allocated to $a_{\mathsf{satisfied}}$, otherwise the utility of $a_{\mathsf{satisfied}}$ would be too low. Reallocating all of these resources is clearly equivalent with finding a satisfying truth-assignment for the formula.

Now we wil describe the reallocation process in a more systematic way:
When the propositional CNF formula denoted by instance $I$ is satisfiable, there is an allocation $a'$ that Pareto-dominates $a$. It can be obtained in the following way. 
\begin{enumerate}
\item Take allocation $a$ and reallocate the resources $\{o_{x_1}, \ldots, o_{x_w}\}$ to the allocation that corresponds to the assignment that satisfies the formula of $I$. By doing this, the utility of $a_{\mathsf{unassigned}}$ becomes lower than the utility it has in allocation $a$. This problem will be dealt with in step 4.
\item By construction, all of the other resources of the agents that obtained a resource in step 1 can now all be reallocated so that the utility of those agents is not decreased below the utility they have in allocation $a$. (The resource they received in step 1 gets them high enough utility to maintain at least the same utility as in $a$, even if they lose all of their other resources.) So we reallocate all those resources `appropriately' to the agents $\{a_{c_1}, \ldots, a_{c_{w'}} \}$. By appropriately we mean that a reallocated resource is reallocated to the single other agent that has non-zero utility for it. By construction, there is precisely one such agent for each item that is reallocated in this step. 
\item Because, in step 2, the utility of agents $\{a_{c_1}, \ldots, a_{c_{w'}} \}$ is increased, we can reallocate the items $\{o_{c_1}, \ldots, o_{c_{w'}}\}$ to agent $a_{\mathsf{satisfied}}$. Without giving the agents $\{a_{c_1}, \ldots, a_{c_{w'}} \}$ a lower utility than in allocation $a$. Now it is the case that each agent except $a_{\mathsf{unassigned}}$ has a utility that is at least as high as allocation $a$. $a_{\mathsf{unassigned}}$ has no items allocated, so his utility is 0. The utility of $a_{\mathsf{satisfied}}$ is $2w'$ in our current allocation, while in allocation $a$ it was $w'$. 
\item So, as a last step, we can reallocate $o_{\mathsf{satisfied}}$ to $a_{\mathsf{unassigned}}$. The utility of $a_{\mathsf{unassigned}}$ is then $w + 1$ in our new allocation $a'$, while it was only $w$ in allocation $a$. By performing this last step, the utility of $a_{\mathsf{satisfied}}$ decreases to $w$, but this is not a problem since the utility of $a_{\mathsf{satisfied}}$ was also $w$ in allocation $a$.
\end{enumerate}
\end{proof}

\begin{lemma}
If the 3-UNSAT instance $I$ is a YES-instance, i.e. the formula is unsatisfiable, then the allocation $a$ in $I'$ is Pareto-optimal.
\end{lemma}
\begin{proof}
In an allocation $a'$ that Pareto-dominates allocation $a$, at least one agent has strictly greater utility in $a'$ than he has in $a$, and all the other agents have a utility that is at least as great. We divide the proof up in cases, and show that in $a'$ no agent can be the agent that has strictly greater utility than he has in $a$, while all other agents don't have a lower utility than they have in $a$.
\begin{description}
\item[Agent $a_{\mathsf{unassigned}}$:] In $a'$, the utility of agent $a_{\mathsf{unassigned}}$ can only be greater than in $a$ if he gets the resource $o_{\mathsf{satisfied}}$. Because the other agents may not have lower utility than they have in $a$, agent $a_{\mathsf{satisfied}}$ needs then be allocated the set of items $\{ o_{c_1}, \ldots, o_{c_{w'}} \}$. By the same argument, every agent $a_{c_i} \in \{ a_{c_1}, \ldots, a_{c_{w'}} \}$ needs to get allocated at least one of the resources $\{ o_{c_i, l} | l \in c_i \}$. If we allocate such a resource $o_{c_i, l}$ to $a_{c_i}$, then the utility of $a_{\mathsf{set}(l)}$ gets too low, and we must compensate by allocating the resource $o_{x_j}, x_j \in l$ to $a_{\mathsf{set}(l)}$. As explained in the previous lemma, regarding $I$ this is equivalent to setting the variable $x_i$ to a truth value such that clause $c_j$ gets satisfied. We must do this for all clauses, so then there must be an assignment where all of the clauses are satisfied, i.e. $I$ must be a satisfiable instance. Which it isn't.
\item[All other cases:] It is also impossible to create an allocation $a'$ that Pareto-dominates $a$, where some agent $a_i \not = a_{\mathsf{unassigned}}$ has strictly greater utility than in $a$, while all the other agents have a utility that is at least as high as the utility that they had in $a$: no matter what agent we choose for the role of $a_i$, it is always neccessary to allocate at least one of the resources in $\{ o_{x_1}, \ldots, o_{x_w}\}$ to an agent other than $a_{\mathsf{unassigned}}$. This means that we are required to allocate $o_{\mathsf{satisfied}}$ to $a_{\mathsf{unassigned}}$, and we fall back to the case we just proved for agent $a_{\mathsf{unassigned}}$. 

It is easy to check that this is true for any $a_i$ that we pick.
\end{description}
\end{proof}
\end{proof}

\section{Complexity of finding an efficient and envy-free allocation for agents with additive utility}\label{section:eef-existence}

The proof given in the previous section was somewhat of an intermediate result that we came across in the process of finding a proof for our next theorem.
We first make an additional definition.

\begin{definition}[Envy-freeness]
Given a resource allocation setting $\langle A = \{a_1, \ldots, a_n\}, \mathcal{O}, U = \{u_1, \ldots, u_n\}, \mathcal{C}, u_c \rangle$ and an admissable allocation $a$, $a$ is called envy-free iff 
\begin{equation*}
\forall a_i \in A : \forall a_j \in A : u_i(a(a_i)) \geq  u_i(a(a_j)).
\end{equation*}

We can define an envy-freeness constraint $c_{\mathsf{envyfree}}$ so that we can add it to $\mathcal{C}$. $a$ then is not admissable if $a$ is not envy-free.

If there exists an $i$ and there exists a $j$ for which $u_i(a(a_j)) > u_i(a(a_i))$ and $i \not = j$, then $a$ is not envy-free and we say that $a_i$ {\it envies} $a_j$ in allocation $a$.
\end{definition}

Now we state the problem and give a proof that this problem is $\Sigma_2^p$-complete.

\begin{definition}[EEF-EXISTENCE-ADDITIVE]
In the problem EEF-EXISTENCE-ADDITIVE we must decide whether there exists a Pareto-efficient and envy-free admissible allocation in the resource allocation setting $\langle A, \mathcal{O}, U, \mathcal{C}, u_c \rangle$, where
\begin{itemize}
\item $\mathcal{C} = \{ c_{\mathsf{preempt}}, c_{\mathsf{envyfree}}\}$,
\item $\forall u \in U : u$ is an additive utility function.
\end{itemize}

The collective utility function $u_c$ can be disregarded here, so the problem is representable as the 3-tuple $\langle A, \mathcal{O}, V\rangle$. In this 3-tuple, $V = \{v_1, \ldots, v_n\}$ represents the utility functions of $U$. For all $1 \leq i \leq n$, $v_i$ is the representation of $u_i$ as described in definition \ref{definition:additiveutility}.
\end{definition}

\begin{theorem}
EEF-EXISTENCE-ADDITIVE is $\Sigma_2^p$-complete.
\end{theorem}
\begin{proof}
Membership of $\Sigma_2^p$ is easily shown. The problem can be decided by an alternating turing machine that makes 1 alternation and starts in an existential state: In the existential state, an allocation $a$ is guessed, and it is checked if this allocation is envy-free. The turing machine then enters the universal state. In this universal state it is checked for all possible allocations if an allocation Pareto-dominates $a$. If this is not the case, then $a$ is Pareto-efficient and envy-free.

We prove hardness by a Karp reduction from the complement of the problem $\forall \exists$3CNF. $\forall \exists$3CNF is $\Pi_2^p$-complete, that is, complete for the complement of $\Sigma_2^p$. It is perhaps the most well known complete problem in the second level of the polynomial hierarchy. We selected this problem from \cite{polyhierarchy}, a list of complete problems in the polynomial hierarchy. 

An instance of $\forall \exists$3CNF consists of two disjoint sets of propositional variables $X_{\forall} = \{x_{1}^{\forall}, \ldots, x_{|X_{\forall}|}^{\forall}\}$ and $X_{\exists} = \{x_{1}^{\exists}, \ldots, x_{|X_{\exists}|}^{\exists}\}$ and a propositional formula in 3CNF over the variables in $X_{\forall} \cup X_{\exists}$. This propositional formula is represented as the set of clauses $C = \{c_1, \ldots, c_{|C|}\}$. A clause $c_i \in C$ is a set of at most 3 literals. The problem for a $\forall \exists$3CNF-instance is to decide whether for every possible assignment of the variables in $X_{\forall}$, there exists some assignment of the variables of $X_{\exists}$ that makes the formula true\footnote{To remove ambiguity: please note that the assignment of the variables in $X_{\exists}$ needs not be the same for every assignment of the variables in $X_{\forall}$.}.

For this proof we must introduce some additional terminology: given a set of propositional variables, in a {\it partial truth-assignment}, or simply {\it partial assignment} to these variables, only a part of the variables are assigned a truth value, and the other part is left unassigned. Also, given a partial assignment $s$ on a set of propositional variables and a propositional formula on these propositional variables, we say that the formula is {\it satisfiable on $s$} iff we can transform $s$ into a full assignment $s'$ by assigning in $s$ a truth-value to the unassigned variables, such that $s'$ satisfies the formula.

We make a minor assumption on the $\forall \exists$3CNF instances. For every variable $x \in X_{\forall} \cup X_{\exists}$, both the literals $x_i$ and $\neg x_i$ must appear at least once in the formula $C$. Fortunately, this assumption can be made without loss of generality: if we have a $\forall \exists$3CNF instance where the assumption doesn't hold for some variable $x
\in X_{\forall} \cup X_{\exists}$, then we can simply add the tautological clause $\{x, \neg x\}$ to $C$. We make this assumption in order to reduce the complicatedness of our reduction.

In this proof we use the following notational conventions. We will use the symbol $l$ to refer to a literal and we will use for any variable $x_i \in X_{\exists} \cup X_{\forall}$ the symbol $l_{x_i}$ to refer to a literal in which $x_i$ occurs. Also, if we use the notation $\neg l_{x_i}$, then by that we mean the positive literal $x_i$ if $l_{x_i}$ is a negative literal, and we mean the negative literal $\neg x_i$ if $l_{x_i}$ is a positive literal. Lastly, We define the set $C_{l_{x_i}}$ for each literal of each variable $x_i \in X_{\exists} \cup X_{\forall}$ as the set of clauses in which $l_{x_i}$ occurs.

The reduction in this proof resembles the reduction in the proof of theorem \ref{additive}: we reuse a lot of the same ideas and tricks. The reduction for this proof however, is more complex. We have to deal this time with universally quantified variables and envy-freeness. Moreover, we cannot ``set'' an allocation in advance, as we could in the reduction of the proof of theorem \ref{additive}. We will now describe the entire reduction. We advise the reader to work out an example for a small $\forall \exists$3CNF-instance in the table format as we did in the proof of theorem \ref{additive}. This is because we won't give an example in this proof: the table format size of the EEF-EXISTENCE-ADDITIVE instance is too large to put on this sheet, even for small instances.

Given a $\forall \exists$3CNF-instance 
\begin{equation*}
I = \langle X_{\forall} = \{x_{1}^{\forall}, \ldots, x_{|X_{\forall}|}^{\forall}\}, X_{\exists} = \{x_{1}^{\exists}, \ldots, x_{|X_{\exists}|}^{\exists}\}, C = \{c_1, \ldots, c_{|C|}\}\rangle,
\end{equation*} 
we reduce it to a EEF-EXISTENCE-ADDITIVE-instance $I' = \langle A, \mathcal{O}, V \rangle$ in the following way.
\begin{itemize}
\item $|A| = 4|X_{\forall}| + 2|X_{\exists}| + |C| + L_{\forall} + 3$, where $L_{\forall}$ is the total number of literal occurences in $C$ of variables in $X_{\forall}$. For each variable $x_i^{\forall} \in X_{\forall}$, four agents $a_{\mathsf{set}(x_i^{\forall})}$, $a_{\mathsf{set}(\neg x_i^{\forall})}$, $a_{\mathsf{set}(x_i^{\forall})}^{\mathsf{helper}}$ and $a_{\mathsf{set}(\neg x_i^{\forall})}^{\mathsf{helper}}$ are introduced. For each variable $x_i^{\exists} \in X_{\exists}$, two agents $a_{\mathsf{set}(x_i^{\exists})}$ and $a_{\mathsf{set}(\neg x_i^{\exists})}$ are introduced. For each clause $c_i \in C$, the agent $a_{c_i}$ is introduced. For all $c_i \in C$, for each literal $l \in c_i$ wherein a variable of $X_{\forall}$ occurs, we introduce the agent $a_{c_i,l}^{\mathsf{envyprotection}}$. The remaining three agents are $a_{\mathsf{unassigned}}$, $a_{\mathsf{unassigned}}^{\mathsf{envyprotection}}$, and $a_{\mathsf{satisfied}}$. 

For ease of explaining and understanding the rest of the proof, we introduce the following symbols and terminology: 
\begin{itemize}
\item We refer to the set $\{a_{c_1}, \ldots, a_{c_{|C|}} \}$ as $A_{\mathsf{ca}}$. Alternatively, we may refer to those agents as {\it clause agents}. 
\item We refer to the set $\{ a_{\mathsf{set}(l)} | x_i^{\exists} \in l \}$ as $A_{\mathsf{evaa}}$ Alternatively, we may refer to those agents as {\it existential variable assignment agents}.
\item We refer to the set $\{ a_{\mathsf{set}(l)} | x_i^{\forall} \in l \}$ as $A_{\mathsf{uvaa}}$ Alternatively, we may refer to those agents as {\it universal variable assignment agents}.
\item We refer to the set $\{ a_{\mathsf{set}(l)}^{\mathsf{helper}} | x_i^{\forall} \in l \}$ as $A_{\mathsf{uvaha}}$ Alternatively, we may refer to those agents as {\it universal variable assignment helper agents}.
\item We refer to the set $\{ a_{c,l}^{\mathsf{envyprotection}} | c \in C \wedge l \in c \}$ as $A_{\mathsf{ulepa}}$. Alternatively, we may refer to those resources as {\it universal literal envy-protection agents}.
% agents: U, Uenvy, sat clause agents, variable assignment agents = (existential variable assignment agents, universal variable assignment agents), universal variable assignment helper agents.
\end{itemize}
Using these definitions, we have 
\begin{equation*}
A = A_{\mathsf{ca}} \cup A_{\mathsf{evaa}} \cup A_{\mathsf{uvaa}} \cup A_{\mathsf{uvaha}} \cup A_{\mathsf{ulepa}} \cup \{a_{\mathsf{unassigned}}, a_{\mathsf{unassigned}}^{\mathsf{envyprotection}}, a_{\mathsf{satisfied}} \}.
\end{equation*}
\item $|\mathcal{O}| = 4|X_{\forall}| + |X_{\exists}| + 2|C| + L + L_{\forall} + 3$, where $L$ is the total number of literal occurences in the 3CNF formula $C$, and $L_{\forall}$ is the total number of literal occurences in $C$ of variables in $X_{\forall}$. For all variables $x_i^{\forall} \in X_{\forall}$, we introduce the resources $o_{x_i^{\forall}}$, $o_{x_i^{\forall}}^{\mathsf{compensation}}$, $o_{\mathsf{set}(x_i^{\forall})}^{\mathsf{helper}}$ and $o_{\mathsf{set}(\neg x_i^{\forall})}^{\mathsf{helper}}$. For all variables $x_i^{\exists} \in X_{\exists}$, we introduce the resource $o_{x_i^{\exists}}$. For each clause $c_i \in C$, we introduce the resources $o_{c_i}$ and $o_{c_i}^{\mathsf{compensation}}$. 
For all $c_i \in C$, for each literal $l \in c_i$, we introduce the resource $o_{c_i, l}$. For all $c_i \in C$, for each literal $l \in c_i$ wherein a variable of $X_{\forall}$ occurs, we introduce the resource $o_{c_i,l}^{\mathsf{envyprotection}}$. The remaining three resources are $o_{\mathsf{satisfied}}$, $o_{\mathsf{envy1}}$ and $o_{\mathsf{envy2}}$.

For ease of explaining and understanding the rest of the proof, we introduce the following symbols and terminology: 
\begin{itemize}
\item We refer to the set $\{o_{c_1}, \ldots, o_{c_{|C|}} \}$ as $\mathcal{O}_{\mathsf{cr}}$. Alternatively, we may refer to those resources as {\it clause resources}. 
\item We refer to the set $\{o_{c_1}^{\mathsf{compensation}}, \ldots, o_{c_{|C|}}^{\mathsf{compensation}} \}$ as $\mathcal{O}_{\mathsf{ccr}}$. Alternatively, we may refer to those resources as {\it clause compensation resources}.
\item We refer to the set $\{ o_{c, l} | c \in C \wedge l \in c \wedge x_{\forall} \in l \wedge x_{\forall} \in X_{\forall} \}$ as $\mathcal{O}_{\mathsf{ulr}}$. Alternatively, we may refer to those resources as {\it universal literal resources}.
\item We refer to the set $\{ o_{c, l} | c \in C \wedge l \in c \wedge x \in l \wedge x \in X_{\exists} \}$ as $\mathcal{O}_{\mathsf{elr}}$. Alternatively, we may refer to those resources as {\it existential literal resources}.
\item We refer to the set $\mathcal{O}_{\mathsf{ulr}} \cup \mathcal{O}_{\mathsf{elr}}$ as $\mathcal{O}_{\mathsf{lr}}$. Alternatively, we may refer to those resources as {\it literal resources}.
\item We refer to the set $\{ o_{x_1^{\forall}} , \ldots , o_{x_{|X_{\forall}|}^{\forall}} \}$ as $\mathcal{O}_{\mathsf{uvr}}$. Alternatively, we may refer to those resources as {\it universal variable resources}.
\item We refer to the set $\{ o_{x_1^{\exists}} , \ldots , o_{x_{|X_{\exists}|}^{\exists}} \}$ as $\mathcal{O}_{\mathsf{evr}}$. Alternatively, we may refer to those resources as {\it existential variable resources}.
\item We refer to the set $\mathcal{O}_{\mathsf{uvr}} \cup \mathcal{O}_{\mathsf{evr}}$ as $\mathcal{O}_{\mathsf{vr}}$. Alternatively, we may refer to those resources as {\it variable resources}.
\item We refer to the set $\{ o_{x_1^{\forall}}^{\mathsf{compensation}} , \ldots , o_{x_{|X_{\forall}|}^{\forall}}^{\mathsf{compensation}} \} $ as $\mathcal{O}_{\mathsf{uvcr}}$. Alternatively, we may refer to those resources as {\it universal variable compensation resources}.
\item We refer to the set $\{ o_{\mathsf{set}(x_i^{\forall})}^{\mathsf{helper}}, \ldots, o_{\mathsf{set}(x_{|X_{\forall}|}^{\forall})}^{\mathsf{helper}} \} \cup \{ o_{\mathsf{set}(\neg x_i^{\forall})}^{\mathsf{helper}}, \ldots, o_{\mathsf{set}(\neg x_{|X_{\forall}|}^{\forall})}^{\mathsf{helper}} \}$ as $\mathcal{O}_{\mathsf{uvahr}}$. Alternatively, we may refer to those resources as {\it universal variable assignment helper resources}.
\item We refer to the set $\{ o_{c,l}^{\mathsf{envyprotection}} | c \in C \wedge l \in c \}$ as $\mathcal{O}_{\mathsf{ulepr}}$. Alternatively, we may refer to those resources as {\it universal literal envy-protection resources}.
% resources: literal resources (universal and existential), clause resources, envy-protection resources for the universal variable clause-satisfying helper agents, envy protection resources for the clause agents, universal variable assignment helper resources, literal resources, sat, envy1, envy2.
\end{itemize}
Using these definitions, we have 
\begin{equation*}
\mathcal{O} = 
\mathcal{O}_{\mathsf{cr}} \cup 
\mathcal{O}_{\mathsf{ccr}} \cup 
\mathcal{O}_{\mathsf{ulr}} \cup 
\mathcal{O}_{\mathsf{elr}} \cup 
\mathcal{O}_{\mathsf{uvr}} \cup 
\mathcal{O}_{\mathsf{evr}} \cup 
\mathcal{O}_{\mathsf{uvcr}} \cup 
\mathcal{O}_{\mathsf{uvahr}} \cup 
\mathcal{O}_{\mathsf{ulepr}} \cup 
\{ o_{\mathsf{satisfied}}, o_{\mathsf{envy1}}, o_{\mathsf{envy2}} \}.
\end{equation*}
\item To complete the reduction, we specify the additive utility functions. Due to the extensive use of subscripts and superscripts for the agents and resources, we don't use the same notation for this as we did in the proof for theorem \ref{additive}. All members of $V$ are vectors of coefficients. $v_i \in V$ is the vector representing the additive utility function of agent $a_i$. The members of $v_i$ are coefficients. In $v_i$ there is one coefficient for each resource in $\mathcal{O}$. We name these coefficients as follows. Let $a \in A$, and let $o \in \mathcal{O}$. Then we simply denote the utility-coefficient of agent $a$ for resource $o$ as $\alpha[a, o]$.

In the list below, let $M$ be an extremely large number. By default all coefficients of all agents are set to zero, with the following exceptions:
\begin{itemize}
\item For all $o_{c_i} \in \mathcal{O}_{\mathsf{cr}}$: 
\begin{eqnarray*}
\alpha[a_{c_i}, o_{c_i}] & := & M, \\ \alpha[a_{\mathsf{satisfied}}, o_{c_i}] & := & 1, \\
\forall l \in c_i : \alpha[a_{c_i,l}^{\mathsf{envyprotection}}, o_{c_i}] & := & M.
\end{eqnarray*}
\item For all $o_{c_i}^{\mathsf{compensation}} \in \mathcal{O}_{\mathsf{ccr}}$: 
\begin{eqnarray*}
\alpha[a_{c_i}, o_{c_i}^{\mathsf{compensation}}] & := & M-1, \\ \alpha[a_{\mathsf{unassigned}}, o_{c_i}^{\mathsf{compensation}}] & := & 1.
\end{eqnarray*}
\item For all $o_{c, l} \in \mathcal{O}_{\mathsf{ulr}}$:
\begin{eqnarray*}
\alpha[a_{c}, o_{c,l}] & := & 1, \\ \alpha[a_{\mathsf{set}(l)}^{\mathsf{helper}}, o_{c,l}] & := & 1, \\
\alpha[a_{c,l}^{\mathsf{envyprotection}}, o_{c,l}] & := & 1.
\end{eqnarray*}
\item For all $o_{c, l} \in \mathcal{O}_{\mathsf{elr}}$:
\begin{eqnarray*}
\alpha[a_{c}, o_{c,l}] & := & 1, \\ \alpha[a_{\mathsf{set}(l)}, o_{c,l}] & := & 1.
\end{eqnarray*}
\item For all $o_{x_i^{\forall}} \in \mathcal{O}_{\mathsf{uvr}}$:
\begin{eqnarray*}
\alpha[a_{\mathsf{set}(x_i^{\forall})}, o_{x_i^{\forall}}] & := & 1, \\ \alpha[a_{\mathsf{set}(\neg x_i^{\forall})},  o_{x_i^{\forall}}] & := & 1.
\end{eqnarray*}
\item For all $o_{x_i^{\exists}} \in \mathcal{O}_{\mathsf{evr}}$:
\begin{eqnarray*}
\alpha[a_{\mathsf{set}(x_i^{\exists})}, o_{x_i^{\exists}}] & := & |C_{x_i^{\exists}}|, \\ \alpha[a_{\mathsf{set}(\neg x_i^{\exists})}, o_{x_i^{\exists}}] & := & |C_{\neg x_i^{\exists}}|, \\ \alpha[a_{\mathsf{unassigned}}, o_{x_i^{\exists}}] & := & 1.
\end{eqnarray*}
\item For all $o_{x_i^{\forall}}^{\mathsf{compensation}} \in \mathcal{O}_{\mathsf{uvcr}}$:
\begin{eqnarray*}
\alpha[a_{\mathsf{set}(x_i^{\forall})}, o_{x_i^{\forall}}^{\mathsf{compensation}}] & := & 1, \\ \alpha[a_{\mathsf{set}(\neg x_i^{\forall})}, o_{x_i^{\forall}}^{\mathsf{compensation}}] & := & 1, \\ \alpha[a_{\mathsf{unassigned}}, o_{x_i^{\forall}}^{\mathsf{compensation}}] & := & 1.
\end{eqnarray*}
\item For all $o_{\mathsf{set}(l_{x_i^{\forall}})}^{\mathsf{helper}} \in \mathcal{O}_{\mathsf{uvahr}}$:
\begin{eqnarray*}
\alpha[a_{\mathsf{set}(l_{x_i^{\forall}})}^{\mathsf{helper}}, o_{\mathsf{set}(l_{x_i^{\forall}})}^{\mathsf{helper}}] & := & |C_{l_{x_i^{\forall}}}|, \\ \alpha[a_{\mathsf{set}(l_{x_i^{\forall}})}, o_{\mathsf{set}(l_{x_i^{\forall}})}^{\mathsf{helper}}] & := & 1, \\ \alpha[a_{\mathsf{set}(\neg l_{x_i^{\forall}})}, o_{\mathsf{set}(l_{x_i^{\forall}})}^{\mathsf{helper}}] & := & 1.
\end{eqnarray*}
\item For all $o_{c,l}^{\mathsf{envyprotection}} \in \mathcal{O}_{\mathsf{ulepr}}$:
\begin{eqnarray*}
\alpha[a_{c,l}^{\mathsf{envyprotection}}, o_{c,l}^{\mathsf{envyprotection}}] & := & M.
\end{eqnarray*}
\item For $o_{\mathsf{satisfied}}$:
\begin{eqnarray*}
\alpha[a_{\mathsf{unassigned}}, o_{\mathsf{satisfied}}] & := & |X_{\exists}| + |X_{\forall}| + |C| + 1 , \\ \alpha[a_{\mathsf{satisfied}}, o_{\mathsf{satisfied}}] & := & |C|.
\end{eqnarray*}
\item For $o_{\mathsf{envy1}}$:
\begin{eqnarray*}
\alpha[a_{\mathsf{unassigned}}, o_{\mathsf{envy1}}] & := & 2 \times \alpha[a_{\mathsf{unassigned}}, o_{\mathsf{satisfied}}], \\ \alpha[a_{\mathsf{satisfied}}, o_{\mathsf{envy1}}] & := & \frac{1}{2}.
\end{eqnarray*}
\item For $o_{\mathsf{envy2}}$:
\begin{eqnarray*}
\alpha[a_{\mathsf{unassigned}}, o_{\mathsf{envy2}}] & := & \alpha[a_{\mathsf{unassigned}}, o_{\mathsf{envy1}}] + |X_{\exists}| + |X_{\forall}| + |C|, \\ \alpha[a_{\mathsf{unassigned}}^{\mathsf{envyprotection}}, o_{\mathsf{envy2}}] & := & M.
\end{eqnarray*}
\end{itemize}
\end{itemize}

That completes the reduction. It should be obvious that generating this EEF-EXISTENCE-ADDITIVE-instance from the $\forall \exists$3CNF instance takes polynomial time. We now continue with the correctness proof.

$\forall \exists$3CNF is a $\Pi_2^p$-complete problem, and we want to prove EEF-EXISTENCE-ADDITIVE is $\Sigma_2^p$-complete. Therefore we need to show that in $I'$ there is only a Pareto-efficient, envy-free (EEF) allocation if there exists some assignment to the variables in $X_{\forall}$ for which there is no assignment to the variables in $X_{\exists}$ which makes the 3CNF-formula $C$ true.

Now we will outline the correctness-proof for this reduction. After that we finish the proof by giving the definition and lemmas that are ommitted in the outline. 

We define in definition \ref{XA-allocations} the specific set of allocations for $I'$, that correspond to a specific type of partial truth-assignment to the variables in $I$. Namely, assignments that satisfy the following two conditions: 
\begin{enumerate}
\item All universally quantified variables are set to either true or false, and 
\item all existential variables are left unassigned. 
\end{enumerate}
In lemma \ref{correctness1} we prove that all allocations that correspond to such truth-assignments are envy-free. We call these allocations $X_\forall$-allocations. We will show in lemma \ref{correctness2} that in $I'$, any EEF allocation must be an $X_\forall$-allocation. Next, we will show in lemmas \ref{correctness3} and \ref{correctness4} that for an $X_{\forall}$-allocation, a Pareto-improvement is possible only if in $I$ the formula can get satisfied on the partial truth-assignment that corresponds to this $X_\forall$-allocation. Now if $I$ is a YES-instance of $\forall \exists$3CNF, then clearly the formula is satisfiable on all partial assignments with the two aforementioned conditions, hence a pareto-improvement is possible on all envy-free allocations. So then $I'$ is a NO-instance of EEF-EXISTENCE-ADDITIVE. On the other hand, if $I$ is a NO-instance of $\forall \exists$3CNF, then clearly there must be a partial assignment satisfying the 2 aforementioned conditions for which the formula is not satisfiable. Hence there is in this case an envy-free allocation that is pareto-optimal. The remainder of the proof consists of definition \ref{XA-allocations} and lemmas \ref{correctness1}, \ref{correctness2}, \ref{correctness3} and \ref{correctness4}.

\begin{definition}[$X_{\forall}$-assignments and $X_{\forall}$-allocations (corrected)]
\label{XA-allocations}
For $I$, we define an {\it $X_{\forall}$-assignment} as a partial assignment to the variables in $X_{\forall} \cup X_{\exists}$ where all variables in $X_{\forall}$ are set to either true or false, and all variables in $X_{\exists}$ are not assigned to a truth value. Given an $X_{\forall}$-assignment $s$, we define the corresponding {\it $X_{\forall}$-allocation} in the following way:
\begin{enumerate}
\item All agents $a_{c_i} \in A_{\mathsf{ca}}$ get allocated the resource $o_{c_i}$.
\item For all $x_i^{\exists} \in X_{\exists}$, all agents $a_{\mathsf{set}(l_{x_i^{\exists}})} \in A_{\mathsf{evaa}}$ get allocated the resources $\{o_{c, l_{x_i^{\exists}}} | l_{x_i^{\exists}} \in c\}$.
\item For all $x_i^{\forall} \in X_{\forall}$, for all pairs of agents $a_{\mathsf{set}(x_i^{\forall})} \in A_{\mathsf{uvaa}}, a_{\mathsf{set}(\neg x_i^{\forall})} \in A_{\mathsf{uvaa}}$. Allocate $o_{x_i^{\forall}}$ to one of the two agents, it doesn't matter which one, say $a_{\mathsf{set}(l_{x_i^{\forall}})}$. Now, if $x_i$ is true in $s$, allocate $o_{\mathsf{set}(x_i^{\forall})}^{\mathsf{helper}}$ to $a_{\mathsf{set}(\neg l_{x_i^{\forall}})}$ and allocate $o_{\mathsf{set}(\neg x_i^{\forall})}^{\mathsf{helper}}$ to $a_{\mathsf{set}(\neg l_{x_i^{\forall}})}^{\mathsf{helper}}$. Otherwise, if $x_i$ is false in $s$, allocate these two resources the other way around: allocate $o_{\mathsf{set}(x_i^{\forall})}^{\mathsf{helper}}$ to $a_{\mathsf{set}(\neg l_{x_i^{\forall}})}^{\mathsf{helper}}$ and allocate $o_{\mathsf{set}(\neg x_i^{\forall})}^{\mathsf{helper}}$ to $a_{\mathsf{set}(\neg l_{x_i^{\forall}})}$. 
\item All agents $a_{c,l}^{\mathsf{envyprotection}} \in A_{\mathsf{ulepa}}$ get the resource $o_{c,l}^{\mathsf{envyprotection}}$.
\item $a_{\mathsf{unassigned}}$ gets allocated all of the resources $\mathcal{O}_{\mathsf{evr}} \cup \mathcal{O}_{\mathsf{ccr}} \cup \mathcal{O}_{\mathsf{uvcr}} \cup \{o_{\mathsf{envy1}} \}$.
\item $a_{\mathsf{unassigned}}^{\mathsf{envyprotection}}$ gets allocated the resource $o_{\mathsf{envy2}}$
\item $a_{\mathsf{satisfied}}$ gets allocated the resource $o_{\mathsf{satisfied}}$.
\item 
\item The only resources that have not been allocated up to this point are the universal literal resources $o_{c,l}$.
If $l$ is not true in $s$, then $o_{c,l}$ can be allocated to either $a_{c,l}^{\mathsf{envyprotection}}$ or they are allocated to $a_{\mathsf{set}(l)}^{\mathsf{helper}}$. It doesn't matter which of the two.
If $l$ is true in $s$, then $o_{c,l}$ must be allocated to $a_{\mathsf{set}(l)}^{\mathsf{helper}}$, and thus may not be allocated to $a_{c,l}^{\mathsf{envyprotection}}$.
\end{enumerate}
\end{definition}

\begin{lemma}
\label{correctness1}
All $X_\forall$-allocations are envy-free.
\end{lemma}
\begin{proof}
Let $a$ be any $X_\forall$-allocation for $I'$ and let $s$ be the corresponding $X_\forall$-assignment for $I$. For every agent we will show that he doesn't envy any other agent. In this proof we say that an agent {\it wants} a resource if the agent has a non-zero utility-coefficient for that resource. For simplicity we also say that an agent {\it has} a resource if he is allocated that resource.
\begin{itemize}
\item $a_{\mathsf{unassigned}}^{\mathsf{envyprotection}}$ doesn't envy any agent because he has the single resource for which he has a non-zero utility-coefficient. 
\item $a_{\mathsf{satisfied}}$ doesn't envy any agent. Its utility in allocation $a$ is $|C|$; the total utility of the $|C|+1$ resources that he wants but doesn't have is $|C|+1$. For all of these $|C|+1$ resources, $a_{\mathsf{satisfied}}$ has a utility-coefficient of 1. So $a_{\mathsf{satisfied}}$ would only envy an agent if there is an agent in $a$ that has all of these $|C|+1$ resources, and that's not the case.
\item $a_{\mathsf{unassigned}}$ doesn't envy any other agent because the only items he wants but doesn't have are $o_{\mathsf{envy2}}$ and $o_{\mathsf{satisfied}}$. The former is allocated to $a_{\mathsf{unassigned}}^{\mathsf{envyprotection}}$ and the latter is allocated to $a_{\mathsf{satisfied}}$. $a_{\mathsf{unassigned}}$ doesn't envy $a_{\mathsf{unassigned}}^{\mathsf{envyprotection}}$ because the utility-coefficient that $a_{\mathsf{unassigned}}$ has for $o_{\mathsf{envy2}}$ is equal to (and not higher than) the utility that $a_{\mathsf{unassigned}}$ currently has in $a$. $a_{\mathsf{unassigned}}$ also doesn't envy $a_{\mathsf{satisfied}}$ because the utility-coefficient that $a_{\mathsf{unassigned}}$ has for $o_{\mathsf{satisfied}}$ is lower than the utility that $a_{\mathsf{unassigned}}$ currently has in $a$.
\item For all $a_{c,l}^{\mathsf{envyprotection}} \in A_{\mathsf{ulepa}}$: $a_{c,l}^{\mathsf{envyprotection}}$ has an item for which he has a utility coefficient of $M$. For $a_{c,l}^{\mathsf{envyprotection}}$, there are two more items that he wants. For one of those items he has a utility-coefficient of $M$. For the other item he has a utility-coefficient of $1$. These items are not both allocated to the same agent, so $a_{c,l}^{\mathsf{envyprotection}}$ envies no-one. 
\item All $a_{c_i} \in A_{\mathsf{ca}}$ have no envy: $a_{c_i}$ has a utility of $M$. The total utility of all items that $a_{c_i}$ wants but doesn't have is $M-1+|c_i|$. For the resource $o_{c_i}^{\mathsf{compensation}}$, $a_{c_i}$ has a utility coefficient of $M-1$. For the other resources that $a_{c_i}$ wants but doesn't have (at most 3), $a_{c_i}$ has a utility coefficient of 1. These are literal resources. Literal resources and $o_{c_i}^{\mathsf{compensation}}$ are not all allocated to the same agent in allocation $a$, so $a_{c_i}$ doesn't envy any agent.
\item For all $a_{\mathsf{set}(l)} \in A_{\mathsf{evaa}}$, $a_{\mathsf{set}(l)}$ has a utility of $|C_{l}|$ in $a$. The maximal utility they can have is $2|C_{l}|$, so $a_{\mathsf{set}(l)}$ doesn't envy anyone because he already has half of his total possible utility. 
\item For all $a_{\mathsf{set}(l)}^{\mathsf{helper}} \in A_{\mathsf{uvaha}}$, $a_{\mathsf{set}(l)}^{\mathsf{helper}}$ has a utility of at least $|C_{l}|$ in $a$. The maximal utility they can have is $2|C_{l}|$, so $a_{\mathsf{set}(l)}$ doesn't envy anyone because he already has half of his total possible utility. 
\item All $a_{\mathsf{set}(l)} \in A_{\mathsf{uvaa}}$ have a utility of 1 in $a$. The maximal utility they can have is 4. There are 3 items that $a_{\mathsf{set}(l)}$ wants but doesn't have. For all of these 3 items, $a_{\mathsf{set}(l)}$ has a utility-coefficient of 1. $a_{\mathsf{set}(l)}$ doesn't envy anyone because each of these 3 items is allocated to a different agent: one of these 3 items is allocated to $a_{\mathsf{unassigned}}$, one is allocated to $a_{\mathsf{set}(\neg l)}$, and one is allocated to either $a_{\mathsf{set}(l)}^{\mathsf{helper}}$ or $a_{\mathsf{set}(\neg l)}^{\mathsf{helper}}$.
\end{itemize}
\end{proof}

\begin{lemma}
\label{correctness2}
All EEF-allocations must be $X_\forall$-allocations.
\end{lemma}
\begin{proof}
We show this by reasoning about how the resources must be allocated in order to achieve envy-freeness and Pareto-optimality. After having done this, it turns out that the set of allocations that are possibly EEF is exactly the set of all $X_{\forall}$-allocations.

First of all, it doesn't make sense to allocate a resource to an agent whose utility-coefficient is zero for that resource. A Pareto-improvement is always possible in such an allocation, by simply reallocating the resource to an agent that has a positive utility-coefficient for it. This is why we will only consider allocating resources to agents who have positive utility-coefficients for the resources. By this argument it immediately follows that all $o_{c,l}^{\mathsf{envyprotection}} \in \mathcal{O}_{\mathsf{ulepr}}$ must be allocated to $a_{c,l}^{\mathsf{envyprotection}}$.

$o_{\mathsf{envy2}}$ must be allocated to $a_{\mathsf{unassigned}}^{\mathsf{envyprotection}}$, or else he would envy agent $a_{\mathsf{unassigned}}$. Also, we see that $a_{\mathsf{unassigned}}$ always envies $a_{\mathsf{satisfied}}$ if $o_{\mathsf{envy1}}$ isn't allocated to $a_{\mathsf{unassigned}}$, because $a_{\mathsf{unassigned}}$ has a utility-coefficient of $2(X_{\exists} + X_{\forall}) + 2$ for $o_{\mathsf{envy1}}$. This is more than half of the maximal utility it is still able to get (given that $o_{\mathsf{envy2}}$ is allocated to $a_{\mathsf{unassigned}}^{\mathsf{envyprotection}}$).

Next, it follows that $o_{\mathsf{satisfied}}$ must be allocated to $a_{\mathsf{satisfied}}$, since if it would be allocated to $a_{\mathsf{unassigned}}$, then $a_{\mathsf{satisfied}}$ always envies $a_{\mathsf{unassigned}}$ because $a_{\mathsf{unassigned}}$ then has the items $o_{\mathsf{satisfied}}$ and $o_{\mathsf{envy1}}$. If $a_{\mathsf{satisfied}}$ would get this bundle of items, then he has a utility that's more than half of his total possible utility, so $a_{\mathsf{satisfied}}$ would envy $a_{\mathsf{unassigned}}$ in that case. 

Given our current set of EEF-allocation-requirements up till now, it's clear that $a_{\mathsf{unassigned}}$ must get allocated all of the resources $\mathcal{O}_{\mathsf{evr}} \cup \mathcal{O}_{\mathsf{uvcr}} \cup \mathcal{O}_{\mathsf{ccr}}$. Only if we allocate all of these resources to $a_{\mathsf{unassigned}}$, then the utility of $a_{\mathsf{unassigned}}$ is high enough to not envy $a_{\mathsf{unassigned}}^{\mathsf{envyprotection}}$. 

At this point, it is certain that for all $o_{c_i} \in \mathcal{O}_{\mathsf{cr}}$, $o_{c_i}$ must be allocated to $a_{c_i}$. This must be the case because: firstly, $a_{c_i}$ has a utility-coefficient of $M$ for this resource; secondly, $a_{c_i}$ has a utility of $M-1$ for $o_{c_i}^{\mathsf{compensation}}$, but according to our current set of EEF-allocation-requirements, $o_{c_i}^{\mathsf{compensation}}$ must already be allocated to $a_{\mathsf{unassigned}}$; and thirdly, $a_{c_i}$ has a utility-coefficient of 1 for all other resources that $a_{c_i}$ wants. That is very low compared to $M$, so even if $a_{c_i}$ would get all of these resources instead of $o_{c_i}$, $a_{c_i}$ would still envy the agent that gets $o_{c_i}$.

Because all items $o_{x_i^{\exists}} \in \mathcal{O}_{\mathsf{evr}}$ must be allocated to $a_{\mathsf{unassigned}}$, the agents $a_{\mathsf{set}(x_i^{\exists})}$ must get allocated all of the resources that $a_{\mathsf{set}(x_i^{\exists})}$ wants, except for $o_{x_i^{\exists}}$. These are exactly the set of resources $\{o_{c,l} | x_i \in l \}$. Allocating these resources to $a_{\mathsf{set}(x_i^{\exists})}$ makes his utility equal to $\alpha[a_{\mathsf{set}(x_i^{\exists})}, o_{x_i^{\exists}}]$, and therefore it is ensured that $a_{\mathsf{set}(x_i^{\exists})}$ doesn't envy anyone. Analogous reasoning holds for the agents $a_{\mathsf{set}(\neg x_i^{\exists})}$: They must get allocated all of the resources that $a_{\mathsf{set}(\neg x_i^{\exists})}$ wants, except for $o_{x_i^{\exists}}$. Allocating these resources to $a_{\mathsf{set}(\neg x_i^{\exists})}$ makes his utility equal to $\alpha[a_{\mathsf{set}(\neg x_i^{\exists})}, o_{x_i^{\exists}}]$, and therefore it is ensured that $a_{\mathsf{set}(\neg x_i^{\exists})}$ doesn't envy anyone.

For all pairs of universal variable assignment agents $a_{\mathsf{set}(x_i)}$ and $a_{\mathsf{set}(\neg x_i)}$, we have the following situation: the total possible utility that both agents can get is 4: they both have four resources that they want, and they both have a utility of 1 for each resource. Also they both want exactly the same four resources. However, we already concluded that the resources $o_{\mathsf{set}(x_i)}^{\mathsf{compensation}}$ and $o_{\mathsf{set}(\neg x_i)}^{\mathsf{compensation}}$ must be allocated to $a_{\mathsf{unassigned}}$. According to this requirement, the total possible utility that both agents can still get is 3. $a_{\mathsf{set}(x_i)}$ and $a_{\mathsf{set}(\neg x_i)}$ are the only agents that can have a positive utility-coefficient for the resource $o_{x_i^{\forall}}$, so we can only allocate this resource to one of these two agents. If we allocate it to either agent, say $a_{\mathsf{set}(l_{x_i})}$, then the other agent $a_{\mathsf{set}(\neg l_{x_i})}$ will envy $a_{\mathsf{set}(l_{x_i})}$ unless he gets allocated one of the other two resources that are left ($x_{\mathsf{set}(x_i)}^{\mathsf{helper}}$ and $x_{\mathsf{set}(\neg x_i)}^{\mathsf{helper}}$). We can choose either one to allocate to $a_{\mathsf{set}(\neg l_{x_i})}$. After we have done this, our only possibility is to allocate the other resource to $a_{\mathsf{set}(\neg l_{x_i})}^{\mathsf{helper}}$ (if we allocate it to $a_{\mathsf{set}(l_{x_i})}$ or $a_{\mathsf{set}(\neg l_{x_i})}$ then there will be envy among $a_{\mathsf{set}(l_{x_i})}$ and $a_{\mathsf{set}(\neg l_{x_i})}$).

For the universal literal resources, the following holds. A universal literal resource $o_{c,l_{x_i^{\forall}}}$ must be allocated to $a_{\mathsf{set}(x_i^{\forall})}^{\mathsf{helper}}$ if $o_{\mathsf{set}(x_i^{\forall})}^{\mathsf{helper}}$ is not assigned to $a_{\mathsf{set}(l_{x_i^{\forall}})}^{\mathsf{helper}}$, or else $a_{\mathsf{set}(l_{x_i^{\forall}})}^{\mathsf{helper}}$ will envy either $a_{\mathsf{set}(x_i^{\forall})}$ or $a_{\mathsf{set}(\neg x_i^{\forall})}$. In the case that $o_{\mathsf{set}(x_i^{\forall})}^{\mathsf{helper}}$ is assigned to $a_{\mathsf{set}(l_{x_i^{\forall}})}^{\mathsf{helper}}$, we have the possibility to allocate $o_{c,l_{x_i^{\forall}}}$ to one of the agents in $\{ a_c, a_{c,l_{x_i^{\forall}}}^{\mathsf{envyprotection}}, a_{\mathsf{set}(l_{x_i^{\forall}})}^{\mathsf{helper}}\}$. But if we would allocate $o_{c,l_{x_i^{\forall}}}$ to $a_c$, then $a_{c,l_{x_i^{\forall}}}^{\mathsf{envyprotection}}$ would envy $a_c$ because $a_c$ has the bundle of items $\{o_c, o_{c,l_{x_i^{\forall}}}\}$. Having this bundle would give $M+1$ to $a_{c,l_{x_i^{\forall}}}^{\mathsf{envyprotection}}$, and $a_{c,l_{x_i^{\forall}}}^{\mathsf{envyprotection}}$ has currently only $M$ utility. So we cannot allocate $o_{c,l_{x_i^{\forall}}}$ to $a_c$, and the only possibilities left are to assign $o_{c,l_{x_i^{\forall}}}$ to either $a_{c,l_{x_i^{\forall}}}^{\mathsf{envyprotection}}$ or $a_{\mathsf{set}(l_{x_i^{\forall}})}^{\mathsf{helper}}$.

The requirements we just described clearly restrict the set of allocations that are possibly EEF, to the set of $X_{\forall}$-allocations.
\end{proof}

\begin{lemma}
\label{correctness3}
Given an $X_{\forall}$-assignment $s$ for $I$, and the $X_{\forall}$-allocation $a$ in $I'$ that corresponds to $s$. If the propositional 3CNF-formula $C$ is satisfiable on $s$, then there is an allocation $a'$ that Pareto-dominates $a$.
\end{lemma}
\begin{proof}
Let $s$ be the $X_{\forall}$-assignment and $a$ be the corresponding $X_{\forall}$-allocation. 
Given $a$, it is possible to reallocate some resources to yield a Pareto-dominating allocation $a'$ where the utility of $a_{\mathsf{unassigned}}$ is increased, and the utility of the other agents is at least as high as in $a$.

First note that the only way to increase the utility of $a_{\mathsf{unassigned}}$ is to reallocate the resource $o_{\mathsf{satisfied}}$ from $a_{\mathsf{satisfied}}$ to $a_{\mathsf{unassigned}}$. If this happens, then $a_{\mathsf{unassigned}}$ gets $|X_{\forall}| + |X_{\exists}| + |C| + 1$ extra utility, so in that case $a_{\mathsf{unassigned}}$ can lose $|X_{\forall}| + |X_{\exists}| + |C|$ utility, and he will still have higher utility than in $a$. We can only move $o_{\mathsf{satisfied}}$ to $a_{\mathsf{unassigned}}$ if we reallocate all of the clause resources to $a_{\mathsf{satisfied}}$, otherwise the utility of $a_{\mathsf{satisfied}}$ would be too low. If we reallocate all of these clause resources, then all clause agents would lose $M$ utility. We can compensate this by reallocating all of the clause compensation resources to the clause agents (this gives $M-1$ utility to each clause agent). There are two problems with this move: first of all, by doing this, $a_{\mathsf{unassigned}}$ loses $|C|$ utility; and secondly each clause resource only gets $M-1$ utility, so we need to allocate each clause resource at least 1 more utility in order to compensate for the loss of $M$ utility of each clause agent. The first problem turns out not to be a problem at all, because $a_{\mathsf{unassigned}}$ has a ``surplus'' of $|X_{\forall}| + |X_{\exists}| + |C|$ utility, and by reallocating all clause compensation resources, $a_{\mathsf{unassigned}}$ loses only $|C|$ utility, so $a_{\mathsf{unassigned}}$ is still allowed to lose $|X_{\forall}| + |X_{\exists}|$ utility. The second problem can be remedied by reallocating at least 1 literal resource to each clause agent. A clause agent $a_{c_i}$ has a utility-coefficient of 1 for a literal resource $o_{c_i,l}$ and a utility-coefficient of 0 for all other literal resources. Literal resources can be either existential literal resources or universal literal resources:
\begin{enumerate}
\item For any universal variable we can execute the following procedure. {\bf Step 1:} for all $o_{x_i^{\forall}}^{\mathsf{compensation}} \in \mathcal{O}_{\mathsf{uvcr}}$, if $x_i^{\forall}$ is assigned to true (false) in $s$, reallocate $o_{x_i^{\forall}}^{\mathsf{compensation}}$ from $a_{\mathsf{unassigned}}$ to the universal variable assignment agent that has currently got the resource $o_{\mathsf{set}(x_i)}^{\mathsf{helper}}$ ($o_{\mathsf{set}(\neg x_i)}^{\mathsf{helper}}$). $a_{\mathsf{unassigned}}$ loses $|X_{\forall}|$ utility by this move, so there is still $|X_{\exists}|$ utility to ``spend'' for $a_{\mathsf{unassigned}}$. {\bf Step 2:} if $x_i^{\forall}$ is assigned to true (false) in $s$, reallocate the item $o_{\mathsf{set}(x_i)}^{\mathsf{helper}}$ ($o_{\mathsf{set}(\neg x_i)}^{\mathsf{helper}}$) from $a_{\mathsf{set}(x_i)}$ to $a_{\mathsf{set}(x_i)}^{\mathsf{helper}}$ ($a_{\mathsf{set}(\neg x_i)}^{\mathsf{helper}}$). Note that by executing steps 1 and 2, no universal variable assignment agent loses any utility. {\bf Step 3:} if $x_i^{\forall}$ is assigned to true (false) in $s$, then for all of the literal resources $o_{c, x_i^{\forall}} \in \{o_{c,x_i^{\forall}} | x_i \in c \}$ ($o_{c, \neg x_i^{\forall}} \in \{o_{c,\neg x_i^{\forall}} | x_i^{\forall} \in c \}$), we move $o_{c, x_i^{\forall}}$ ($o_{c, \neg x_i^{\forall}}$)from $a_{\mathsf{set}(x_i^{\forall})}^{\mathsf{helper}}$ ($a_{\mathsf{set}(\neg x_i^{\forall})}^{\mathsf{helper}}$) to $a_c$. Note that by this last step, no universal variable assignment helper agent loses any utility. 

When we execute the procedure we just described, we can move a certain set of clause resources to $a_{\mathsf{satisfied}}$ without lowering anyone's utility. By construction, this set of clause resources corresponds exactly the set of clauses that are satisfied by the $X_{\forall}$-assignment $s$. The clause resources that correspond to clauses that still need to get satisfied, still need to get reallocated. We will see how to do this in the next step:
\item For any existential variable $x_i^{\exists}$, an existential literal resource $o_{c,l_{x_i^{\exists}}}$ is allocated in $a$ to an existential variable assignment agent $a_{\mathsf{set}(l_{x_i^{\exists}})}$. If we reallocate $o_{c,l_{x_i^{\exists}}}$ to $a_{c}$, then we have to compensate this by moving an $o_{x_i}$ to $a_{\mathsf{set}(x_i^{\exists})}$. In $a$, $o_{x_i}$ is allocated to $a_{\mathsf{unassigned}}$. So if we reallocate all of the existential variable resources, $a_{\mathsf{unassigned}}$ loses $|X_{\exists}|$ utility. So after reallocating the existential variable resources, $a_{\mathsf{unassigned}}$ may not lose any utility anymore, since we still want $a_{\mathsf{unassigned}}$ to have strictly greater utility than in $a$. We can reallocate an existential variable resource $o_{x_i^{\exists}}$ only to $a_{\mathsf{set}(x_i^{\exists})}$ or $a_{\mathsf{set}(\neg x_i^{\exists})}$. Altogether this means that for any existential variable $x_i^{\exists}$, we can give extra utility to either the clause agents $\{a_{c} | x_i \in c \}$ or to the clause agents $\{a_{c} | \neg x_i \in c \}$. Remember that this extra utility is needed for the clause agents in order to be able to reallocate the clause resources to $a_{\mathsf{satisfied}}$. In this sense, reallocating all of the clause resources to $a_{\mathsf{satisfied}}$ is equivalent to finding a truth-assignment for the unassigned variables of $s$ such that $C$ is satisfied. Such an assignment exists by our assumption, hence from $a$, a Pareto-improvement to $a'$ is possible. $a'$ is clearly not EEF because it is not an $X_{\forall}$-allocation. (More concretely, in $a'$, $a_{\mathsf{satisfied}}$ envies $a_{\mathsf{unassigned}}$ because $a_{\mathsf{unassigned}}$ has the bundle of items $\{ a_{\mathsf{envy1}}, a_{\mathsf{satisfied}} \}$.)
\end{enumerate}
\end{proof}

\begin{lemma}
\label{correctness4}
Given an $X_{\forall}$-assignment $s$ for $I$, and the $X_{\forall}$-allocation $a$ in $I'$ that corresponds to $s$. If the propositional 3CNF-formula $C$ is unsatisfiable on $s$, then $a$ is EEF.
\end{lemma}
\begin{proof}
By lemma \ref{correctness1}, $a$ is envy-free, so we only need to show that $a$ is Pareto-optimal.
We do this by proving the following two things: 
\begin{enumerate}
\item There doesn't exist an allocation $a'$ that Pareto-dominates $a$ in which all clause-resources are allocated to $a_{\mathsf{satisfied}}$.
\item Any allocation $a'$ that Pareto-dominates $a$ must have all clause-resources allocated to $a_{\mathsf{satisfied}}$.
\paragraph{Proof for 1:}
This is a lot like our story in the previous lemma. We will try to make a Pareto-dominating allocation $a'$ where all clause resources are allocated to $a_{\mathsf{satisfied}}$. We do this by trying to transform $a$ into $a'$, and we will see that this is not possible. 

If we take $a$, and reallocate the clause resources to $a_{\mathsf{satisfied}}$, then all of the clause agents lose $M$ utility. The only way to compensate is reallocating all of the clause compensation resources to the clause agents and reallocating to every clause agent at least one literal resource. If we reallocate the clause compensation resources then the utility of $a_{\mathsf{unassigned}}$ is lowered by $|C|$ and needs to be compensated. The only way to do so is to reallocate $o_{\mathsf{satisfied}}$ to $a_{\mathsf{unassigned}}$. This is no problem: the utility of $a_{\mathsf{satisfied}}$ was $|C|$ in allocation $a$, and now it is still $|C|$. 

The reallocation of at least one literal resource to every clause agent is going to be the problem. There are literal resources allocated to four types of agents:
\begin{itemize}
\item Some universal literal resources may in $a$ be allocated to universal literal envy-protection agents. It is impossible to reallocate such literals because it is impossible to compensate the utility of these agents by giving them another resource. The only resources these agents want but don't have are the clause resources, but the are already reallocated to $a_{\mathsf{satisfied}}$.
\item Some universal literal resources may in $a$ be allocated to universal variable assignment helper agents $a_{\mathsf{set}(l)}^{\mathsf{helper}}$ who already have resource $o_{\mathsf{set}(l)}^{\mathsf{helper}}$. It is impossible to reallocate these literal resources: the only resources that $a_{\mathsf{set}(l)}^{\mathsf{helper}}$ wants but doesn't have can be literal resources that are allocated to a universal literal envy-protection agent. We can not reallocate these literal resources, as we argued in the previous item of this list.
\item Some universal literal resources may in $a$ be allocated to universal variable assignment helper agents $a_{\mathsf{set}(l)}^{\mathsf{helper}}$ who do not have resource $o_{\mathsf{set}(l)}^{\mathsf{helper}}$. In this case, it is possible to reallocate these literal resources to the clause agents. The only way to compensate the loss of utility of agent $a_{\mathsf{set}(l)}^{\mathsf{helper}}$ by reallocating the resource $o_{\mathsf{set}(l)}^{\mathsf{helper}}$ from one of the two universal variable assignment agents to $a_{\mathsf{set}(l)}^{\mathsf{helper}}$. Subsequently we can compensate the loss of the universal variable assignment agent by reallocating a universal variable compensation resource from $a_{\mathsf{unassigned}}$ to him. $a_{\mathsf{unassigned}}$ may lose all of its universal variable compensation resources. Its utility will still remain higher than it was in $a$ because it has received the resource $o_{\mathsf{satisfied}}$.

Just as in the previous lemma, the procedure we just mentioned will add at least 1 extra utility to a certain set of clause agents. This set of clause agents are the clause agents that correspond to the clauses that are satisfied by the partial truth-assignment $s$. 
\item The existential literal resources are allocated to the existential variable assignment agents. It is possible to reallocate some of these existential literal resources to the clause agents. As we already pointed out in the previous lemma, for any existential variable $x_i^{\exists}$, we can give extra utility to either the clause agents $\{a_{c} | x_i \in c \}$ or to the clause agents $\{a_{c} | \neg x_i \in c \}$. 
\end{itemize}

So, we can give a literal to the clause agents that correspond to clauses satisfied by $s$. And the remaining clause agents we can give a literal if it is possible to reallocate an existential literal resource to these clause agent. This is obviously equivalent to finding a truth-assignment to the variables in $X_{\exists}$ that satisfies formula $C$ on $s$. By our assumption such a truth-assignment doesn't exist, so there exists no allocation $a'$ that Pareto-dominates $a$ in which all clause-resources are allocated to $a_{\mathsf{satisfied}}$.
\paragraph{Proof for 2:}
For each agent, we show that we can only transform $a$ to a Pareto-dominating allocation $a'$ and increase that agent's utility if we allocate all clause-resources to $a_{\mathsf{satisfied}}$.
\begin{description}
\item[For agent $a_{\mathsf{unassigned}}$:] The only way to improve the utility of $a_{\mathsf{unassigned}}$ is to reallocate $o_{\mathsf{satisfied}}$ from $a_{\mathsf{satisfied}}$ to $a_{\mathsf{unassigned}}$. But then sat would lose $|C|$ utility. The only way to remedy this is to reallocate all of the $|C|$ clause-resources to $a_{\mathsf{satisfied}}$.
\item[For all existential variable assignment assignment agents $a_{\mathsf{set}(l_{x_i^{\exists}})}$:] The only way to increase the utility of $a_{\mathsf{set}(l_{x_i^{\exists}})}$ is to reallocate $o_{x_i^{\exists}}$ from $a_{\mathsf{unassigned}}$ to $a_{\mathsf{set}(l_{x_i^{\exists}})}$. Now, $a_{\mathsf{unassigned}}$ loses 1 utility, so we need to increase $a_{\mathsf{unassigned}}$'s utility by allocating him the resource $o_{\mathsf{satisfied}}$. So we fall back to the case for $a_{\mathsf{unassigned}}$.
\item[For all universal variable assignment agents $a_{\mathsf{set}(l_{x_i^{\forall}})}$:] Reallocating $a_{\mathsf{set}(x_i^{\forall})}^{\mathsf{compensation}}$ to $a_{\mathsf{set}(l_{x_i^{\forall}})}$ is not possible. In that scenario we would again fall back to the case for $a_{\mathsf{unassigned}}$. Reallocating an item from $a_{\mathsf{set}(\neg l_{x_i^{\forall}})}$ to $a_{\mathsf{set}(l_{x_i^{\forall}})}$ does not help. This action removes 1 utility from $a_{\mathsf{set}(\neg l_{x_i^{\forall}})}$. Hence we would need to be able to increase the utility of $a_{\mathsf{set}(\neg l_{x_i^{\forall}})}$, but $a_{\mathsf{set}(l_{x_i^{\forall}})}$ and $a_{\mathsf{set}(\neg l_{x_i^{\forall}})}$ have exactly the same utility coefficients, i.e. they are clones of each other. So, needing to improve the utility of $a_{\mathsf{set}(\neg l_{x_i^{\forall}})}$ is the same problem as needing to improve the utility of $a_{\mathsf{set}(l_{x_i^{\forall}})}$. 

We can try one more thing to improve the utility of $a_{\mathsf{set}(l_{x_i^{\forall}})}$ or $a_{\mathsf{set}(l_{x_i^{\forall}})}$. Given that $x_i^{\forall}$ is true in $s$ (if $x_i^{\forall}$ is false in $s$, the reasoning is analogous), we can try to increase the utility of $a_{\mathsf{set}(l_{x_i^{\forall}})}$ or $a_{\mathsf{set}(\neg l_{x_i^{\forall}})}$ by reallocating to either agent the resource $o_{\mathsf{set}(\neg x_i^{\forall})}^{\mathsf{helper}}$ from $a_{\mathsf{set}(\neg x_i^{\forall})}^{\mathsf{helper}}$. Because we remove half of the total possible utility of $a_{\mathsf{set}(\neg x_i^{\forall})}^{\mathsf{helper}}$ with this move, this move can only possibly be done if in $a$, $o_{\mathsf{set}(\neg x_i^{\forall})}^{\mathsf{helper}}$ is the only resource that $a_{\mathsf{set}(\neg x_i^{\forall})}^{\mathsf{helper}}$ has. But even in this case it will turn out that it's impossible: it is easily seen that it would require reallocating a clause resource to a universal literal envy-protection agent, without lowering anyone's utility below the utility he has in allocation $a$. We will show that this is not possible when we arrive at the case for the universal literal envy-protection agents.
\item[For all universal variable assignment helper agents $a_{\mathsf{set}(l_{x_i^{\forall}})}^{\mathsf{helper}}$:] There are two cases here: either $a_{\mathsf{set}(l_{x_i^{\forall}})}^{\mathsf{helper}}$ has resource $o_{\mathsf{set}(l_{x_i^{\forall}})}^{\mathsf{helper}}$ or $a_{\mathsf{set}(l_{x_i^{\forall}})}^{\mathsf{helper}}$ doesn't have resource $o_{\mathsf{set}(l_{x_i^{\forall}})}^{\mathsf{helper}}$. In the former case, it is possible to try to increase the utility of $a_{\mathsf{set}(l_{x_i^{\forall}})}^{\mathsf{helper}}$ by reallocating the resource $o_{\mathsf{set}(l_{x_i^{\forall}})}^{\mathsf{helper}}$ to $a_{\mathsf{set}(l_{x_i^{\forall}})}^{\mathsf{helper}}$. If we do this, then $a_{\mathsf{set}(l_{x_i^{\forall}})}$ gets into trouble and we have to increase his utility. Therefore we fall back to the case for $a_{\mathsf{set}(l_{x_i^{\forall}})}$.

In the second case we can only try to increase the utility of $a_{\mathsf{set}(l_{x_i^{\forall}})}^{\mathsf{helper}}$ by allocating to him a literal resource for which he has a non-zero utility-coefficient. If there is such a literal resource, then it is allocated in $a$ to a universal literal envy-protection agent. This would require reallocating a clause resource to a universal literal envy-protection agent, without lowering anyone's utility below the utility he has in allocation $a$. We will show that this is not possible when we arrive the case for the universal literal envy-protection agents.
\item[For all clause agents $a_{c_i}$:] The utility of a clause agent $a_{c_i}$ can only be improved by reallocating 1 or more of the literal-resources to him for which $a_{c_i}$ has non-zero utility. Let $o_{c_i, l}$ be this literal resource. If we reallocate $o_{c_i, l}$ to $a_{c_i}$, then we would in turn need to improve the utility of an existential variable assignment agent, universal variable assignment helper agent or a universal literal envy-protection agent. For the first two, we refer back to their cases, that we already handled in this list. For the last one, the universal literal envy-protection agent, we will show that it's impossible to increase his utility. We arrive at this case now:
\item[For all universal literal envy-protection agents $a_{c,l}^{\mathsf{envyprotection}}$:] The only way to increase the utility of $a_{c,l}^{\mathsf{envyprotection}}$ is to reallocate the clause resource $o_c$ from $a_c$ to $a_{c,l}^{\mathsf{envyprotection}}$. By this move, $a_c$ would lose $M$ utility. To compensate it we need at least to allocate $o_{c}^{\mathsf{compensation}}$ from $a_{\mathsf{unassigned}}$ to $a_c$. This implies that we will need to increase the utility of $a_{\mathsf{unassigned}}$. From the case we already handled for agent $a_{\mathsf{unassigned}}$, we conclude that we would need to assign all of the clause resources to $a_{\mathsf{satisfied}}$. 
\item[For agent $a_{\mathsf{satisfied}}$:] We can try to reallocate one or more clause-resources to $a_{\mathsf{satisfied}}$. If we do that, then we need to improve the utility of at least one clause agent. As shown as a previous case in this list, improving the utility of this clause agent implies that we need to move all of the clause resources to $a_{\mathsf{satisfied}}$. Another possibility is to try to reallocate $a_{\mathsf{envy1}}$ to $a_{\mathsf{satisfied}}$. But then we would need to improve the utility of $a_{\mathsf{unassigned}}$. As shown, this implies that we would need to move all of the clause resources to $a_{\mathsf{satisfied}}$.
\item[For agent $a_{\mathsf{unassigned}}^{\mathsf{envyprotection}}$:] Obviously we can not improve the utility for this agent, because in $a$ it already has gotten allocated the single resource that he wants.
\end{description} 
\end{enumerate} 
\end{proof}
\end{proof}

\bibliographystyle{plain}
\bibliography{taskallocation}

\end{document}